\documentclass[a4paper,10pt]{article}
\pdfoutput=1

\usepackage{RR,a4wide}

\usepackage[utf8x]{inputenc}
\usepackage{figlatex, subfigure}
\usepackage[american]{babel}
\graphicspath{{fig/},{graphs/}}

\usepackage[ruled,vlined]{algorithm2e_hacked}
\usepackage{color}
\usepackage{rotating}
\usepackage{xspace, url}
\usepackage{amsmath,amssymb, amsthm}
 \newtheorem{definition}{Definition}
 \newtheorem{theorem}{Theorem}
 \newtheorem{lemma}{Lemma}

\usepackage{paralist}
\usepackage{multirow}
\usepackage[shadow]{todonotes}

\usepackage{tikz}
\usetikzlibrary{trees,shapes}

\newcommand{\f}[1]{\ensuremath{\mathit{f}_{#1}}\xspace}
\newcommand{\n}[1]{\ensuremath{\mathit{n}_{#1}}\xspace}
\newcommand{\len}[1]{\ensuremath{\mathit{w}_{#1}}\xspace}
\newcommand{\W}[1]{\ensuremath{\mathit{W}_{#1}}\xspace}

\newcommand{\parent}[1]{\ensuremath{\mathit{parent}({#1})}\xspace}
\newcommand{\Children}[1]{\ensuremath{\mathit{Children}({#1})}\xspace}
\newcommand{\inputs}[1]{\ensuremath{\mathit{inputs}({#1})}\xspace}
\newcommand{\Done}{\ensuremath{\mathit{Done}}\xspace}
\newcommand{\Ancestors}[1]{\ensuremath{\mathit{Ancestors}(#1)}\xspace}

\newcommand{\false}{\ensuremath{\mathit{false}}\xspace}
\newcommand{\true}{\ensuremath{\mathit{true}}\xspace}

\newcommand{\ParSubtrees}{\textsc{ParSubtrees}\@\xspace}
\newcommand{\ParSubtreesMath}{\textsc{ParSubtrees}} 
\newcommand{\ParSubtreesOptim}{\textsc{ParSubtreesOptim}\@\xspace}
\newcommand{\SplitSubtrees}{\textsc{SplitSubtrees}\@\xspace}
\newcommand{\ParInnerFirst}{\textsc{ParInnerFirst}\@\xspace}
\newcommand{\ParInnerFirstMemLimit}{\textsc{ParInnerFirstMemLimit}\@\xspace}
\newcommand{\ParInnerFirstMemLimitOptim}{\textsc{ParInnerFirstMemLimitOptim}\@\xspace}

\newcommand{\ParDeepestFirst}{\textsc{ParDeepestFirst}\@\xspace}
\newcommand{\ParDeepestFirstMemLimit}{\textsc{ParDeepestFirstMemLimit}\@\xspace}
\newcommand{\ParDeepestFirstMemLimitOptim}{\textsc{ParDeepestFirstMemLimitOptim}\@\xspace}
\newcommand{\BookingInnerFirst}{\textsc{MemBookingInnerFirst}\@\xspace}

\newcommand{\PO}{\ensuremath{\mathit{PO}}}
\definecolor{mygray}{gray}{0.75}

\RRdate{October 2014}
\RRauthor{Lionel Eyraud-Dubois\footnote{INRIA, France} \and Loris Marchal\footnote{CNRS, France} \and Oliver
  Sinnen\footnote{Univ. of Auckland, New Zealand} \and Fr\'ed\'eric
  Vivien\footnote{INRIA, France}}
\authorhead{L. Eyraud-Dubois, L. Marchal, O. Sinnen and F. Vivien}
\RRtitle{Ordonnancement parallèle d'arbres de tâches avec mémoire limitée}
\RRetitle{Parallel scheduling of task trees with limited memory}
\titlehead{Parallel scheduling of task trees with limited memory}

\RRresume{Dans ce rapport, nous nous intéressons au traitement
  d'arbres de tâches par plusieurs processeurs. Chaque arête d'un tel
  arbre représente un gros fichier d'entrée/sortie. Une tâche peut
  être traitée seulement si l'ensemble de ses fichiers d'entrée et de
  sortie peut résider en mémoire, et un fichier ne peut être retiré de
  la mémoire que lorsqu'il a été traité. De tels arbres surviennent, par
  exemple, lors de la factorisation de matrices creuses par des
  méthodes multifrontales. La quantité de mémoire nécessaire dépend de
  l'ordre de traitement des tâches. Avec un seul processeur,
  l'objectif est naturellement de minimiser la quantité de mémoire
  requise. Ce problème a déjà été étudié et des algorithmes
  polynomiaux ont été proposés.

  Nous étendons ce problème en considérant plusieurs processeurs, ce
  qui est d'un intérêt évident pour le problème de la factorisation de
  grandes matrices. Avec plusieurs processeurs se pose également le
  problème de la minimisation du temps nécessaire pour traiter
  l'arbre. Nous montrons que comme attendu, ce problème est bien plus
  compliqué que dans le cas séquentiel. Nous étudions la complexité de
  ce problème et nous fournissons des résultats d'inaproximabilité,
  même dans le cas de poids unitaires. Nous proposons plusieurs
  heuristiques qui obtiennent différents compromis entre mémoire et
  temps d'exécution. Certaines d'entre elles sont capables de traiter
  l'arbre tout en gardant la consommation mémoire inférieure à une
  limite donnée.  Nous analysons les performances de toutes ces
  heuristiques par une large campagne de simulations utilisant des
  arbres réalistes.  }

\RRabstract{  This paper investigates the execution of tree-shaped task graphs
  using multiple processors. Each edge of such a tree represents some
  large data. A task can only be executed if all input and output data
  fit into memory, and a data can only be removed from memory after
  the completion of the task that uses it as an input data. Such trees
  arise, for instance, in the multifrontal method of sparse matrix
  factorization. The peak memory needed for the processing of the
  entire tree depends on the execution order of the tasks. With one
  processor the objective of the tree traversal is to minimize the
  required memory. This problem was well studied and optimal
  polynomial algorithms were proposed.

  Here, we extend the problem by considering multiple processors,
  which is of obvious interest in the application area of matrix
  factorization. With multiple processors comes the additional
  objective to minimize the time needed to traverse the tree, i.e., to
  minimize the makespan. Not surprisingly, this problem proves to be
  much harder than the sequential one. We study the computational
  complexity of this problem and provide inapproximability results
  even for unit weight trees. We design a series of practical
  heuristics achieving different trade-offs between the minimization
  of peak memory usage and makespan. Some of these heuristics are able
  to process a tree while keeping the memory usage under a given
  memory limit.  The different heuristics are evaluated in an
  extensive experimental evaluation using realistic trees.
}

\RRmotcle{Algorithmes d'approximation, consommation mémoire, optimisation multi-critères, ordonnancement, graphe de tâches}
\RRkeyword{Approximation algorithms, memory usage, multi-criteria optimization, pebble-game, scheduling, task graphs}
\RRprojet{ROMA}
\RCGrenoble

\RRNo{8606}

\begin{document}

\makeRR

\selectlanguage{american}

\section{Introduction}

Parallel workloads are often modeled as task graphs, where nodes
represent tasks and edges represent the dependencies between
tasks. There is an abundant literature on task graph scheduling when
the objective is to minimize the total completion time, or
makespan. However, with the increase of the size of the data to be
processed, the memory footprint of the application can have a dramatic
impact on the algorithm execution time, and thus needs to be
optimized.  This is best exemplified with an application which,
depending on the way it is scheduled, will either fit in the memory,
or will require the use of swap mechanisms or out-of-core techniques.
There are very few existing studies on the minimization of the memory
footprint when scheduling task graphs, and even fewer of them
targeting parallel systems.


We consider the following memory-aware parallel scheduling problem for
rooted trees. The nodes of the tree correspond to tasks, and the edges
correspond to the dependencies among the tasks. The dependencies are
in the form of input and output files\footnote{The concept of \emph{file}
is used here in a very general meaning and does not necessarily correspond to
a classical file on a disk. Essentially, a file is a set of data.}: each
node takes as input several large files, one for each of its children,
and it produces a single large file, and the different files may have
different sizes. Furthermore, the execution of any node requires its
\emph{execution} file to be present; the execution file models
the program and/or the temporary data of the task.
We are to execute such a set of tasks on a parallel
system made of $p$ identical processing resources sharing the same
memory. The execution scheme corresponds to a schedule of the tree
where processing a node of the tree translates into reading the
associated input files and producing the output file.  How can the tree be
scheduled so as to optimize the memory usage?

Modern computing platforms exhibit a complex memory hierarchy ranging
from caches to RAM and disks and even sometimes tape storage, with the
classical property that the smaller the memory, the quicker. Thus, to
avoid large running times, one usually wants to avoid the use of
memory devices whose IO bandwidth is below a given threshold: even if
out-of-core execution (when large data are unloaded to disks) is
possible, this requires special care when programming the application
and one usually wants to stay in the main memory (RAM). This is why in
this paper, we are interested in the question of minimizing the amount
of \emph{main memory} needed to completely process an application.

Throughout the paper, we consider \emph{in-trees} where a task can be
executed only if all its children have already been executed (This is
absolutely equivalent to considering \emph{out-trees} as a solution
for an in-tree can be transformed into a solution for the
corresponding out-tree by just reversing the arrow of time, as outlined
in~\cite{ipdps-tree-traversal}).
%
%
A task can be
processed only if all its
files (input, output, and execution) fit in currently available
memory. At a given time, many files may be stored in the memory, and
at most $p$ tasks may be processed by  the $p$ processors. This is
obviously possible only if all tasks and execution files fit in memory. When a
task finishes, the memory needed for its execution file and its input files
is released.
Clearly, the schedule which determines the processing times of each
task plays a key role in determining which amount of main memory is
needed for a successful execution of the entire tree.

\medskip The first motivation for this work comes from numerical
linear algebra. Tree workflows (assembly or elimination trees) arise
during the factorization of sparse matrices, and the huge size of the
files involved makes it absolutely necessary to reduce the memory
requirement of the factorization.  The sequential version of this
problem (i.e., with $p=1$ processor) has already been studied.
Liu~\cite{Liu86} discusses how to find a memory-minimizing traversal
when the traversal is required to correspond to a postorder traversal
of the tree. A follow-up study~\cite{Liu87} presents an optimal
algorithm to solve the general problem, without the postorder
constraint on the traversal. Postorder traversals are known to be
arbitrarily worse than optimal traversals for memory
minimization~\cite{ipdps-tree-traversal}. However, they are very
natural and straightforward solutions to this problem, as they allow
to fully process one subtree before starting a new one. Therefore,
they are thus widely used in sparse matrix software like
\texttt{MUMPS}~\cite{amestoy2001fam,amestoy2005hsp}, and in practice,
they achieve close to optimal performance on actual elimination
trees~\cite{ipdps-tree-traversal}.


The parallel version of this problem is a natural continuation of
these studies: when processing large elimination trees, it is very
meaningful to take advantage of parallel processing
resources. However, to the best of our knowledge, no theoretical study
exists for this problem. A preliminary version of this work, with
fewer complexity results and proposed heuristics, was presented at
IPDPS 2013~\cite{MarchalSiVi13IPDPS}.  The key contributions of this
work are:
\begin{compactitem}
\item A new proof that the parallel variant of the \emph{pebble game}
  problem is NP-complete (simpler than
  in~\cite{MarchalSiVi13IPDPS}). This shows that the introduction of
  memory constraints, in the simplest cases, suffices to make the
  problem NP-hard (Theorem~\ref{thm.np}).
\item The proof that no schedule can simultaneously achieve a
  constant-ratio approximation for the memory minimization and for the
  makespan minimization (Theorem~\ref{thm.shared.inapprox}); bounds on
  the achievable approximation ratios for makespan and memory when
  the number of processors is fixed (Theorems~\ref{thm.newalphabeta}
  and~\ref{thm.shared.inapprox_chains}).
\item A series of practical heuristics achieving different trade-offs
  between the minimization of peak memory usage and makespan; some of
  these heuristics are guaranteed to keep the memory under a
  given memory limit.
\item An exhaustive set of simulations using realistic tree-shaped
  task graphs corresponding to elimination trees of actual matrices;
  the simulations assess the relative and absolute performance of
  the heuristics.
\end{compactitem}

The rest of this paper is organized as
follows. Section~\ref{sec.related} reviews related studies. The
notation and formalization of the problem are introduced in
Section~\ref{sec.model}. Complexity results are presented in
Section~\ref{sec.complexity} while Section~\ref{sec.heuristics}
proposes different heuristics to solve the problem, which are
evaluated in Section~\ref{sec.experiments}.

\section{Background and Related Work}
\label{sec.related}

\subsection{Sparse matrix factorization}

As mentioned above, determining a memory-efficient tree traversal is
very important in sparse numerical linear algebra. The elimination
tree is a graph theoretical model that represents the storage
requirements, and computational dependencies and requirements, in the
Cholesky and LU factorization of sparse matrices. In a previous study~\cite{ipdps-tree-traversal},
we have described how such trees are built, and how the multifrontal
method~\cite{liu92multifrontal} organizes the computations along the
tree.  This is the context of the founding
studies of Liu~\cite{Liu86,Liu87} on memory minimization for
postorder or general tree traversals presented in the previous
section. Memory minimization is still a concern in modern multifrontal
solvers when dealing with large matrices. Among other, efforts have
been made to design dynamic schedulers that take into account dynamic
pivoting (which impacts the weights of edges and nodes) when
scheduling elimination trees with strong memory
constraints~\cite{guermouche04ipdps}, or to consider both task and tree
parallelism with memory constraints~\cite{agulloPP12}. While these
studies try to optimize memory management in existing parallel
solvers, we aim at designing a simple model to study the fundamental
underlying scheduling problem.

\subsection{Scientific workflows}

The problem of scheduling a task graph under memory constraints also
appears in the processing of scientific workflows whose tasks require
large I/O files. Such workflows arise in many scientific fields, such
as image processing, genomics or geophysical simulations. The problem
of task graphs handling large data has been identified
in~\cite{ramakrishnan07ccgrid} which proposes some simple heuristic
solutions. Surprisingly, in the context of quantum chemistry
computations, Lam et al.~\cite{rauber11CLSS} have recently
rediscovered the algorithm published in 1987 in~\cite{Liu87}.

\subsection{Pebble game and its variants}
\label{sec.pebble}

\newcommand{\MinMemory}{\textsc{MinMemory}\xspace}
\newcommand{\MinIO}{\textsc{MinIO}\xspace}

On the more theoretical side, this work builds upon the many papers
that have addressed the pebble game and its variants.  Scheduling a
graph on one processor with the minimal amount of memory amounts to
revisiting the I/O pebble game with pebbles of arbitrary sizes that
must be loaded into main memory before \emph{firing} (executing) the
task.  The pioneering work of Sethi and Ullman~\cite{SethiUllman70}
deals with a variant of the pebble game that translates into the
simplest instance of our problem when all input/output files have
weight 1 and all execution files have weight 0. The concern
in~\cite{SethiUllman70} was to minimize the number of registers that
are needed to compute an arithmetic expression.  The problem of
determining whether a general DAG can be executed with a given number
of pebbles has been shown NP-hard by Sethi~\cite{Sethi73} if no vertex
is pebbled more than once (the general problem allowing recomputation,
that is, re-pebbling a vertex which has been pebbled before, has been
proven \textsc{Pspace} complete~\cite{gilbert80}). However, this
problem has a polynomial complexity for tree-shaped
graphs~\cite{SethiUllman70}.

To the best of our knowledge, there have been no attempts to
extend these results to parallel machines, with the objective of
minimizing both memory and total execution time. We present such an
extension in Section~\ref{sec.complexity}.


\section{Model and objectives}
\label{sec.model}

\subsection{Application model}

We consider in this paper a tree-shaped task-graph $T$ composed of $n$
nodes, or tasks, numbered from $1$ to $n$. Nodes in the tree have an
output file, an execution file (or program), and several input files
(one per child). More precisely:
\begin{itemize}
\item Each node $i$ in the tree has an execution file of size \n{i}
  and its processing on a processor takes time \len{i}.
\item Each node $i$ has an output file of size \f{i}. If $i$ is not
  the root, its output file is used as input by its parent
  $\parent{i}$; if $i$ is the root, its output file can be of size
  zero, or contain outputs to the outside world.
\item Each non-leaf node $i$ in the tree has one input file per
  child. We denote by $\Children{i}$ the set of the children of
  $i$. For each child $j \in \Children{i}$, task $j$ produces a file
  of size \f{j} for $i$. If $i$ is a leaf-node, then $\Children{i} =
  \emptyset$ and $i$ has no input file: we consider that the initial
  data of the task either resides in its execution file or is read
  from disk (or received from the outside word) during the execution
  of the task.
\end{itemize}

During the processing of a task $i$, the memory must contain its input
files, the execution file, and the output file. The memory needed for
this processing is thus:
$$
\left(\sum_{j \in \Children{i}} \f{j}\right) + \n{i} + \f{i}
$$
After $i$ has been processed, its input files and execution file
(program) are discarded, while its output file is kept in memory
until the processing of its parent.


\subsection{Platform model and objectives}

In this paper, our goal is to design a simple platform model which
allows to study memory minimization on a parallel platform. We thus
consider $p$ identical processors sharing a single memory.

Any sequential optimal schedule for memory minimization is obviously
an optimal schedule for memory minimization on a platform with any
number $p$ of processors. Therefore, memory minimization on parallel
platforms is only meaningful in the scope of multi-criteria
approaches that consider trade-offs between the following two
objectives:
\begin{itemize}
\item \textbf{Makespan:} the classical
  makespan, or total execution time, which corresponds to the
  time-span between the beginning of the execution of the first leaf
  task and the end of the processing of the root task.
\item \textbf{Memory:} the amount of memory
  needed for the computation. At each time step, some files are stored
  in the memory and some task computations occur, inducing a memory
  usage. The \emph{peak memory} is the maximum usage of the memory
  over the whole schedule, hence the memory that needs to be available,
  which we aim to minimize.
\end{itemize}


\section{Complexity results in the pebble game model}
\label{sec.complexity}

Since there are two objectives, the decision version of our problem can be stated as follows.
\begin{definition}[BiObjectiveParallelTreeScheduling]
  Given a tree-shaped task graph $T$ with file sizes and
  task execution times, $p$ processors, and two bounds $B_{C_{\max}}$ and
  $B_{\mathit{mem}}$, is there a schedule of the task graph on the
  processors whose makespan is not larger than $B_{C_{\max}}$ and
  whose peak memory is not larger than $B_{\mathit{mem}}$?
\end{definition}

This problem is obviously NP-complete. Indeed, when there are no
memory constraints ($B_{\mathit{mem}} = \infty$) and when the task
tree does not contain any inner node, that is, when all tasks are
either leaves or the root, then our problem is equivalent to
scheduling independent tasks on a parallel platform which is an
NP-complete problem as soon as tasks have different execution
times~\cite{LenstraRKBr77}. Conversely, minimizing the makespan for a
tree of same-size tasks can be solved in polynomial-time when there
are no memory constraints~\cite{Hu61}. In this section, we consider
the simplest variant of the problem. We assume that all input files
have the same size ($\forall i, \f{i} = 1$) and no extra memory is
needed for computation ($\forall i, \n{i}=0$). Furthermore, we assume
that the processing of each node takes unit time: $\forall i, \len{i}
= 1$. We call this variant of the problem the \emph{Pebble Game} model
since it perfectly corresponds to the pebble game problems introduced
above: the weight $\f{i} = 1$ corresponds to the pebble one must put
on node $i$ to process it; this pebble must remain there until the
parent of node $i$ has been completed, because the parent of node $i$
uses as input the output of node $i$. Processing a node 
is done in unit time.

In this section, we first show that, even in this simple variant, the
introduction of memory constraints (a limit on the number of pebbles) makes
the problem NP-hard (Section~\ref{sec:np}). Then, we show that when
trying to minimize both memory and makespan, it is not possible to get
a solution with a constant approximation ratio for both objectives,
and we provide tighter ratios when the number of processors is fixed
(Section~\ref{sec:inapprox}). 


\subsection{NP-completeness}
\label{sec:np}

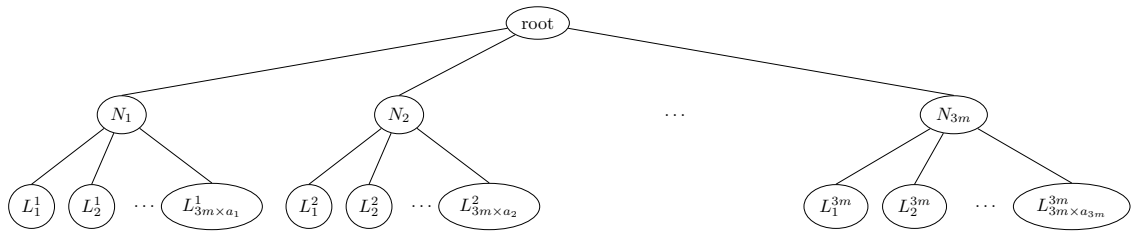
\begin{figure*}
  \centering
  \resizebox{\textwidth}{!}{
    \begin{tikzpicture}[scale=0.9, child anchor = north]
      \tikzstyle{every node}=[ellipse, draw]
      \node{root}
      [sibling distance=60mm, level distance=20mm]
      child{
        node{$N_1$}    [sibling distance=13mm]
        child{node{$L^1_1$}}
        child{node{$L^1_2$}}
        child{node[draw=none]{\ldots~~~} edge from parent [draw=none]}
        child{node{$L^1_{3m\times a_1}$}}
      }
      child{
        node{$N_2$}    [sibling distance=13mm]
        child{node{$L^2_1$}}
        child{node{$L^2_2$}}
        child{node[draw=none]{\ldots~~~} edge from parent [draw=none]}
        child{node{$L^2_{3m\times a_2}$}}
      }
      child{node[draw=none]{\ldots} edge from parent [draw=none]}
      child{
        node{$N_{3m}$}    [sibling distance=17mm]
        child{node{$L^{3m}_1$}}
        child{node{$L^{3m}_2$}}
        child{node[draw=none]{\ldots~~~} edge from parent [draw=none]}
        child{node{$L^{3m}_{3m\times a_{3m}}$}}
      }
      ;
    \end{tikzpicture}
  }
  \caption{Tree used for the NP-completeness proof}
  \label{fig:np}
\end{figure*}

\begin{theorem}
  \label{thm.np}
  The BiObjectiveParallelTreeScheduling problem is NP-complete in the
  Pebble Game model (i.e., with $\forall i, \f{i} = \len{i} = 1, \n{i}
  = 0$).
\end{theorem}

\begin{proof}
  First, it is straightforward to check that the problem is in NP:
  given a schedule, it is easy to compute its peak memory and
  makespan.

  To prove the problem NP-hard, we perform a reduction from
  3-\textsc{Partition}, which is known to be
  NP-complete in the strong sense~\cite{GareyJohnson}. We consider the following instance
  $\mathcal{I}_1$ of the 3-\textsc{Partition} problem: let $a_i$ be $3m$ integers and $B$
  an integer such that $\sum a_i = m B$. We consider the variant of
  the problem, also NP-complete, where $\forall i, B/4 < a_i < B/2$.
  To solve $\mathcal{I}_1$, we need to solve the following question: does there
  exist a partition of the $a_i$'s in $m$ subsets $S_1,\ldots, S_m$,
  each containing exactly 3 elements, such that, for each $S_k$,
  $\sum_{i \in S_k} a_i = B$?  We build the following instance $\mathcal{I}_2$
  of our problem, illustrated in Figure~\ref{fig:np}. The tree
  contains a root $r$ with $3m$ children, the $N_i$'s, each one
  corresponding to a value
  $a_i$. Each node $N_i$ has $3m \times a_i$ children, $L_1^i, ..., L_{3m\times a_i}^i$, which are
  leaf nodes. The question is to find a schedule of this tree on $p=3mB$
  processors, whose peak memory is not larger than $B_{\mathit{mem}} =
  3mB+3m$ and whose makespan is not larger than $B_{C_{\max}} =
  2m+1$.

  Assume first that there exists a solution to $\mathcal{I}_1$, i.e., that there are $m$
  subsets $S_k$ of 3 elements with $\sum_{i \in S_k} a_i = B$. In this
  case, we build the following schedule:
  \begin{itemize}
  \item At step 1, we process all the nodes $L_x^{i_1}$, $L_y^{j_1}$,
    and $L_z^{k_1}$ with $S_1 = \{a_{i_1}, a_{j_1}, a_{k_1}\}$. There are
    $3mB = p$ such nodes, and the amount of memory needed is also $3mB$.
  \item At step 2, we process the nodes $N_{i_1}$, $N_{j_1}$,
    $N_{k_1}$. The memory needed is $3mB+3$.
  \item At step $2n+1$, with $1 \leq n \leq m-1$, we process the $3mB = p$ nodes $L_x^{i_n}$, $L_y^{j_n}$,
    $L_z^{k_n}$ with $S_n = \{a_{i_n}, a_{j_n}, a_{k_n}\}$. 
    The amount of memory needed 
    is
    $3mB+3n$ (counting the memory for the output files of the $N_t$ nodes previously processed).
  \item At step $2n+2$, with $1 \leq n \leq m-1$, we process the  nodes $N_{i_n}$, $N_{j_n}$,
    $N_{k_n}$. The memory needed for this step is $3mB+3(n+1)$.
  \item At step $2m+1$, we process the root node and the memory needed
    is $3m+1$.
  \end{itemize}
  Thus, the peak memory of this schedule is $B_{\mathit{mem}}$ and its
  makespan $B_{C_{\max}}$.

  Reciprocally, assume that there exists a solution to problem
  $\mathcal{I}_2$, that is, there exists a schedule of makespan
  at most $B_{C_{\max}}=2m+1$.  Without loss of generality, we assume
  that the makespan is exactly $2m+1$. We start by proving that at any
  step of the algorithm, at most three of the $N_i$ nodes
  are being processed. By contradiction, assume that four (or more)
  such nodes $N_{i_s}, N_{j_s}, N_{k_s}, N_{l_s}$ are processed during
  a certain step $s$. We recall that $a_i>B/4$ so that
  $a_{i_s}+a_{j_s}+a_{k_s}+a_{l_s} > B$ and thus
  $a_{i_s}+a_{j_s}+a_{k_s}+a_{l_s} \geq B+1$. The memory needed at
  this step is thus at least $(B+1)3m$ for the children of the nodes
  $N_{i_s}$, $N_{j_s}$, $N_{k_s}$, and $N_{l_s}$ and $4$ for the nodes
  themselves, hence a total of at least $(B+1)3m+4$, which is more
  than the prescribed bound $B_{\mathit{mem}}$. Thus, at most three of
  $N_i$ nodes are processed at any step.  In the considered schedule,
  the root node is processed at step $2m+1$. Then, at step $2m$, some
  of the $N_i$ nodes are processed, and at most three of them from
  what precedes. The $a_i$'s corresponding to those nodes make the
  first subset $S_1$. Then all the nodes $L_x^j$ such that $a_j \in
  S_1$ must have been processed at the latest at step $2m-1$, and they
  occupy a memory footprint of $3m\sum_{a_j \in S_1} a_j$ at steps
  $2m-1$ and $2m$. Let us assume that a node $N_k$ is processed at
  step $2m-1$. For the memory bound $B_{\mathit{mem}}$ to be
  satisfied we must have $a_k + \sum_{a_j \in S_1} a_j \leq
  B$. (Otherwise, we would need a memory of at least $3m(B+1)$ for the
  involved $L_x^j$ nodes plus 1 for the node $N_k$). Therefore, node
  $N_k$ can as well be processed at step $2m$ instead of step $2m-1$.
  We then modify the
  schedule so as to schedule $N_k$ at step $2m$ and thus we add $k$ to
  $S_1$. We can therefore assume, without loss of generality, that no
  $N_i$ node is processed at step $2m-1$. Then, at step $2m-1$ only
  the children of the $N_j$ nodes with $a_j \in S_1$ are processed,
  and all of them are. So, none of them have any memory footprint
  before step $2m-1$. We then generalize this analysis: at step $2i$,
  for $1 \leq i \leq m-1$, only some $N_j$ nodes are processed and
  they define a subset $S_i$; at step $2i-1$, for $1 \leq i \leq m-1$,
  are processed exactly the nodes $L_x^k$ that are children of the
  nodes $N_j$ such that $a_j \in S_i$.


  Because of the memory constraint, each of the $m$ subsets of $a_i$'s
  built above sum to at most $B$. Since they contain all $a_i$'s,
  their sum is $mB$. Thus, each subset $S_k$ sums to $B$ and we have
  built a solution for $\mathcal{I}_1$.
\end{proof}

\subsection{Joint minimization of both objectives}

\label{sec:inapprox}

As our problem is NP-complete, it is natural to wonder whether
approximation algorithms can be designed. In this section, we prove
that there does not exist any scheduling algorithm which approximates
both the minimum makespan and the minimum peak memory with constant
factors. This is equivalent to saying that there is no \emph{Zenith}
(also called \emph{simultaneous}) approximation. We first state a
lemma, valid for any tree-shaped task graph, which provides lower
bounds for the makespan of any schedule.

\begin{lemma}\label{le.lower.bounds}
  For any schedule $S$ on $p$ processors with a peak memory $M$, we
  have the two following lower bounds on the makespan $C_{\max}$:
  \begin{align*}
    C_{\max} \ &  \geq \ \frac{1}{p} \sum_{i=1}^n \len{i}\\
    M \times C_{\max} \ &  \geq \ \sum_{i=1}^n \left(\n{i} + \f{i} + \sum_{j \in \Children{i}} \f{j}\right)\len{i}
  \end{align*}
  In the pebble game model, these equations can be written as:
  \begin{align*}
    C_{\max} \ &  \geq \ n/p \\
    M \times C_{\max} \  &\  \geq 2 n-1
  \end{align*}
\end{lemma}

\begin{proof}
  The first inequality is a classical bound that states that all tasks
  must be processed before  $C_{\max}$.


  Similarly, each task $i$ uses a memory of $\n{i} + \f{i} + \sum_{j
    \in \Children{i}} \f{j}$ during a time $\len{i}$. Hence, the total
   memory usage (i.e., the sum over all time instants $t$ of the memory
   used by $S$ at time $t$) needs to be at least equal to $\sum_{i=1}^n
  \left(\n{i} + \f{i} + \sum_{j \in \Children{i}}
    \f{j}\right)\len{i}$. Because $S$ uses a memory that is not larger
  than $M$ at any time, the total memory usage is upper bounded by
  $M\times C_{\max}$. This gives us the second inequality.  In the
  pebble game model, the right-hand-side term of the second inequality
  can be simplified:
  $$\sum_{i=1}^n \left(\n{i} + \f{i} + \sum_{j \in
      \Children{i}} \f{j}\right)\len{i}
  = \left(\sum_{i=1}^n \f{i}
  \right) + \left(\sum_{i=1}^n\sum_{j \in \Children{i}} \f{j}\right)
  = n + (n-1)$$
\end{proof}


In the next theorem, we show that it is not possible
to design an algorithm with constant approximation ratios for both
makespan and maximum memory, \textit{i.e.} approximation ratios independent of
the number of processors $p$. In Theorem~\ref{thm.newalphabeta}, we
will provide a refined version which analyzes the dependence on $p$.

\begin{theorem}\label{thm.shared.inapprox}
  For any given constants $\alpha$ and $\beta$, there does not exist
  any algorithm for the pebble game model that is both an
  $\alpha$-approximation for makespan minimization and a
  $\beta$-approximation for peak memory minimization when scheduling
  in-tree task graphs.
\end{theorem}

\begin{figure*}
  \centering
  \scalebox{0.7}{
  \begin{tikzpicture}[child anchor = north, label distance = -2pt]
    \tikzstyle{every node} = [circle, fill=black, draw]
    \tikzstyle{every label} = [font=\Large,fill=none,draw=none]
    \node (root) [label=above:root]{}
    [sibling distance=40mm, level distance=15mm]
    child{  node[label=left:$a_1$]{}
      [sibling distance=10mm, level distance=15mm]
      child{node[label=below:$b_{1,1}$]{}}
      child{node[label=below:$b_{1,2}$]{}}
      child{node[draw=none, fill=white]{...} edge from parent [draw=none]}
      child{node[label=below:$b_{1,m}$]{}}
    }
    child{  node[label=left:$a_2$]{}
      [sibling distance=10mm, level distance=15mm]
      child{node[label=below:$b_{2,1}$]{}}
      child{node[label=below:$b_{2,2}$]{}}
      child{node[draw=none, fill=white]{...} edge from parent [draw=none]}
      child{node[label=below:$b_{2,m}$]{}}
    }
    child{ node[draw=none, fill=white]{\Huge \bfseries \ldots } edge from parent [draw=none] }
    child {
      node[label=left:$a_m$]{}
      [sibling distance=10mm, level distance=15mm]
      child{node[label=below:$b_{m,1}$]{}}
      child{node[label=below:$b_{m,2}$]{}}
      child{node[draw=none, fill=white]{...} edge from parent [draw=none]}
      child{node[label=below:$b_{m,m}$]{}}
    };
  \end{tikzpicture}}
  \caption{Tree used for establishing Theorem~\ref{thm.shared.inapprox}.}
  \label{fig.new-proof-alpha-beta}
\end{figure*}
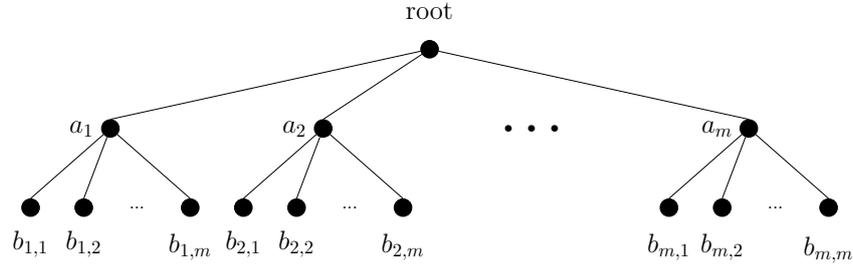

\begin{proof}
  We consider in this proof the tree depicted in
  Figure~\ref{fig.new-proof-alpha-beta}. The root of this tree has $m$
  children $a_1,\ldots,a_m$. Any of these children, $a_i$, has $m$
  children $b_{i,1},\ldots, b_{i,m}$. Therefore, overall this tree
  contains $n= 1 + m + m\times m$ nodes. On the one hand, with a large
  number of processors (namely, $m^2$), this tree can be processed in
  $C^*_{\max}=3$ time steps: all the leaves are processed in the first
  step, all the $a_i$ nodes in the second step, and finally the root
  in the third and last step. On the other hand, the minimum memory
  required to process the tree is $M^* = 2m$. This is achieved by
  processing the tree with a single processor. The subtrees rooted at
  the $a_i$'s are processed one at a time. The processing of the
  subtree rooted at node $a_i$ requires a memory of $m+1$ (for the
  processing of its root $a_i$ once the $m$ leaves have been
  processed). Once such a subtree is processed there is a unit file
  that remains in memory. Hence, the peak memory usage when processing
  the $j$-th of these subtrees is $(j-1)+(m+1)=j+m$. The overall peak
  $M^* = 2m$ is thus reached when processing the root of the last of
  these subtrees.

  Let us assume that there exists a schedule $S$ which is both an
  $\alpha$-approximation for the makespan and a $\beta$-approximation
  for the peak memory. Then, for the tree of
  Figure~\ref{fig.new-proof-alpha-beta}, the makespan $C_{\max}$ of
  $S$ is at most equal to $3\alpha$, and its peak memory $M$ is at
  most equal to $2 \beta m$. Because $n = 1 + m + m^2$,
  Lemma~\ref{le.lower.bounds} implies that $M\times C_{\max} \geq
  2n-1=2m^2+2m+1$. Therefore $M \geq \frac{2m^2+
    2m+1}{3\alpha}$. For a sufficiently large value of $m$, this is
  larger than $2 \beta m$, the upper bound on $M$. This contradicts
  the hypothesis that $S$ is a $\beta$-approximation for peak memory
  usage.
\end{proof}


Theorem~\ref{thm.shared.inapprox} only considers approximation
algorithms whose approximation ratios are constant. In the next
theorem we consider algorithms whose approximations ratios may depend
on the number of processors in the platform.

\begin{theorem}\label{thm.newalphabeta}
  When scheduling in-tree task graphs in the pebble-game model on a
  platform with $p \geq 2$ processors, there does not exist any
  algorithm that is both an $\alpha(p)$-approximation for makespan
  minimization and a $\beta(p)$-approximation for peak memory
  minimization, with
  $$ \alpha(p)\beta(p) \ < \ \frac{2p}{\lceil\log(p)\rceil+2}\cdot$$
\end{theorem}

\begin{proof}
  We establish this result by contradiction. We assume that there
  exists an algorithm that is an $\alpha(p)$-approximation for
  makespan minimization and a $\beta(p)$-approximation for peak memory
  minimization when scheduling in-tree task graphs, with
  $\alpha(p)\beta(p) = \frac{2p}{(\lceil\log(p)\rceil+2)}
  -\epsilon$ with $\epsilon > 0$.

  The proof relies on a tree similar to the one depicted in
  Figure~\ref{fig.newalphabeta} for the case $p=13$. The top part of
  the tree is a complete binary subtree with $\lceil\frac{p}{2}\rceil$ leaves,
  $l_1,\ldots,l_{\lceil\frac{p}{2}\rceil}$, and of height
  $m$. Therefore, this subtree contains $2 \lceil\frac{p}{2}\rceil-1$
  nodes, all of its leaves are at depth either $m$ or $m-1$, and
  $m=1+\lceil\log(\lceil\frac{p}{2}\rceil)\rceil=\lceil\log(p)\rceil$.
  To prove the last equality, we consider whether $p$ is even:
  \begin{compactitem}
  \item $p$ is even: $\exists l\in \mathbb{N}, p=2l$. Then,
    $1+\lceil\log(\lceil\frac{p}{2}\rceil)\rceil=
    \lceil\log(2\lceil\frac{2l}{2}\rceil)\rceil= \lceil\log(2l)\rceil=
    \lceil\log(p)\rceil.  $
  \item $p$ is odd: $\exists l\in \mathbb{N}, p=2l+1$. Then,
    $1+\lceil\log(\lceil\frac{p}{2}\rceil)\rceil=
    \lceil\log(2\lceil\frac{2l+1}{2}\rceil)\rceil=
    \lceil\log(2l+2)\rceil$. Since $2l+1$ is odd,
    $\lceil\log(2l+2)\rceil =
    \lceil\log(2l+1)\rceil=\lceil\log(p)\rceil$.
  \end{compactitem}
  Each node $l_i$ is the root of a \emph{comb} subtree of height $k$
  (except the last node if $p$ is odd); each comb subtree contains
  $2k-1$ nodes. If $p$ is odd, the last leaf of the binary top
  subtree, $l_{\lceil\frac{p}{2}\rceil}$, is the root of a chain
  subtree with $k-1$ nodes. Then, the entire tree contains $n=pk-1$
  nodes (be careful not to count twice the roots of the comb
  subtrees):
  \begin{compactitem}
  \item $p$ is even: $\exists l\in \mathbb{N}, p=2l$. Then,
    $$
    n = \left(2 \left\lceil\frac{p}{2}\right\rceil-1\right) + \left(\left\lceil\frac{p}{2}\right\rceil(2k-2) \right)
    = (2l-1) + l(2k-2)=pk-1.
    $$
  \item $p$ is odd: $\exists l\in \mathbb{N}, p=2l+1$. Then,
    \begin{multline*}
      n = \left(2 \left\lceil\frac{p}{2}\right\rceil-1\right) +
      \left(\left(\left\lceil\frac{p}{2}\right\rceil-1\right)(2k-2)
      \right) + (k-2)\\ = (2(l+1)-1) + (l+1-1)(2k-2)+(k-2) =(2l+1)k-1 =pk-1.
    \end{multline*}
  \end{compactitem}

  \newlength{\combleveldist}
  \setlength\combleveldist{6mm} 
  \newlength{\combsibdist}
  \setlength\combsibdist{4mm} 
  \newlength{\combdist}
  \setlength\combdist{10mm} 
  \newlength{\binleveldisttop}
  \newlength{\binleveldistmiddle}
  \newlength{\binleveldistbottom}
  \setlength\binleveldisttop{10mm} 
  \setlength\binleveldistmiddle{15mm} 
  \newlength{\dotdistance}
  \setlength\dotdistance{\binleveldisttop+\binleveldistmiddle} 
  \newlength{\combrootlevel}
  \setlength{\combrootlevel}{10mm+\dotdistance}
  \newlength{\combleaflevel}
  \setlength{\combleaflevel}{10mm+\dotdistance+5\combleveldist}
  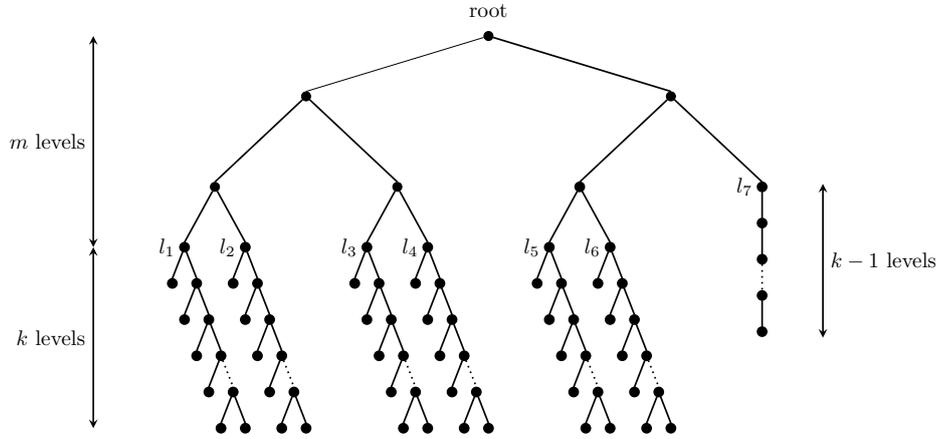
\begin{figure*}
    \centering
    \scalebox{0.8}{
      \begin{tikzpicture}[child anchor = north, label distance = -2pt, rangearrow/.style={thick,<->,>=stealth}]
        \tikzstyle{every node} = [circle, fill=black, draw, inner sep=1.5pt]
        \tikzstyle{every label} = [fill=none, draw=none]
        \node (root) at (0,0) [circle,draw,label=above:root]{}
        [sibling distance=60mm, level distance=\binleveldisttop]
        child{  node[circle, draw]{}
          [sibling distance=30mm, level distance=\binleveldistmiddle]
          child{node[circle, draw]{}
            [sibling distance=\combdist, level distance=\binleveldisttop]
            child{node[circle, solid, draw, label=left:$l_1$]{}
              [sibling distance=\combsibdist, level distance=\combleveldist]
              child{ node {} }
              child{ node {}
                child{ node {} }
                child{ node {}
                  child{ node {} }
                  child{ node {}
                    child{ node {} }
                    child{ node {}
                      child{ node[solid] {} edge from parent[solid]}
                      child{ node[solid] {} edge from parent[solid]}
                      edge from parent[dotted, thick]
                    }
                  }
                }
              }
              edge from parent[solid]}
            child{node[circle, solid, draw, label=left:$l_2$]{}
              [sibling distance=\combsibdist, level distance=\combleveldist]
              child{ node {} }
              child{ node {}
                child{ node {} }
                child{ node {}
                  child{ node {} }
                  child{ node {}
                    child{ node {} }
                    child{ node {}
                      child{ node[solid] {} edge from parent[solid]}
                      child{ node[solid] {} edge from parent[solid]}
                      edge from parent[dotted, thick]
                    }
                  }
                }
              }
              edge from parent[solid]}
            edge from parent [solid, thick]
          }
          child{node[circle, draw]{}
            [sibling distance=\combdist, level distance=\binleveldisttop]
            child{node[circle, solid, draw, label=left:$l_3$]{}
              [sibling distance=\combsibdist, level distance=\combleveldist]
              child{ node {} }
              child{ node {}
                child{ node {} }
                child{ node {}
                  child{ node {} }
                  child{ node {}
                    child{ node {} }
                    child{ node {}
                      child{ node[solid] {} edge from parent[solid]}
                      child{ node[solid] {} edge from parent[solid]}
                      edge from parent[dotted, thick]
                    }
                  }
                }
              }
              edge from parent[solid]}
            child{node[circle, solid, draw, label=left:$l_4$]{}
              [sibling distance=\combsibdist, level distance=\combleveldist]
              child{ node {} }
              child{ node {}
                child{ node {} }
                child{ node {}
                  child{ node {} }
                  child{ node {}
                    child{ node {} }
                    child{ node {}
                      child{ node[solid] {} edge from parent[solid]}
                      child{ node[solid] {} edge from parent[solid]}
                      edge from parent[dotted, thick]
                    }
                  }
                }
              }
              edge from parent[solid]}
            edge from parent [solid, thick]
          }
        }
        child{  node[circle, draw]{}
          [sibling distance=30mm, level distance=\binleveldistmiddle]
          child{node[circle, draw]{}
            [sibling distance=\combdist, level distance=\binleveldisttop]
            child{node[circle, solid, draw, label=left:$l_5$]{}
              [sibling distance=\combsibdist, level distance=\combleveldist]
              child{ node {} }
              child{ node {}
                child{ node {} }
                child{ node {}
                  child{ node {} }
                  child{ node {}
                    child{ node {} }
                    child{ node {}
                      child{ node[solid] {} edge from parent[solid]}
                      child{ node[solid] {} edge from parent[solid]}
                      edge from parent[dotted, thick]
                    }
                  }
                }
              }
              edge from parent[solid]}
            child{node[circle, solid, draw, label=left:$l_6$]{}
              [sibling distance=\combsibdist, level distance=\combleveldist]
              child{ node {} }
              child{ node {}
                child{ node {} }
                child{ node {}
                  child{ node {} }
                  child{ node {}
                    child{ node {} }
                    child{ node {}
                      child{ node[solid] {} edge from parent[solid]}
                      child{ node[solid] {} edge from parent[solid]}
                      edge from parent[dotted, thick]
                    }
                  }
                }
              }
              edge from parent[solid]}
            edge from parent [solid, thick]
          }
          child{node[circle, solid, draw, label=left:$l_7$]{}
              [sibling distance=\combsibdist, level distance=\combleveldist]
              child{ node {}
                child{ node {}
                  child{ node {}
                    child{ node[solid] {} edge from parent[solid]}
                    edge from parent[dotted, thick]
                  }
                }
              }
              edge from parent[solid]}
            edge from parent [solid, thick]
          }
        ;
        \path (-65mm,0mm) edge[rangearrow] node[draw=none, fill=none,        label=left:$m$ levels\ ~] {} (-65mm,-\combrootlevel);
        \path (-65mm,-\combrootlevel) edge[rangearrow] node[draw=none,        fill=none, label=left:$k$ levels\ ~] {} (-65mm,-\combleaflevel);
        \path (55mm,10.5mm-\combrootlevel) edge[rangearrow] node[draw=none,        fill=none, label=right:\ $k-1$ levels] {} (55mm,15mm-\combleaflevel);
      \end{tikzpicture}}
    \caption{Tree used to establish Theorem~\ref{thm.newalphabeta}
      for $p=13$ processors.}
    \label{fig.newalphabeta}
  \end{figure*}

  With the $p$ processors, it is possible to process all comb subtrees
  (and the chain subtree if $p$ is odd) in parallel in $k$ steps by
  using two processors per comb subtree (and one for the chain
  subtree). Then, $m-1$ steps are needed to complete the processing
  of the binary reduction (the $l_i$ nodes have already been processed
  at the last step of the processing of the comb subtrees). Thus, the
  optimal makespan with $p$ processors is $C_{\max}^* =k+m-1$.

  We now compute the optimal peak memory usage, which is obtained with
  a sequential processing. Each comb subtree can be processed with 3
  units of memory, if we follow any postorder traversal starting from
  the deepest leaves.  We consider the sequential processing of the
  entire tree that follows a postorder traversal that process each
  comb subtree as previously described, that process first the
  leftmost comb subtree, then the second leftmost comb subtree, the
  parent node of these subtrees, and so on, and finishes with the
  rightmost comb subtree (or the chain subtree if $p$ is odd). The
  peak memory is reached when processing the last comb subtree. At
  that time, either $m-2$ or $m-1$ edges of the binary subtree are
  stored in memory (depending on the value of
  $\lceil\frac{p}{2}\rceil$). The processing of the last comb subtree
  itself uses 3 units of memory. Hence, the optimal peak memory is not
  greater than $m+2$: $M^* \leq m+2$.

  Let $C_{\max}$ denote the makespan achieved by the studied algorithm
  on the tree, and let $M$ denote its peak memory usage. By
  definition, the studied algorithm is an $\alpha(p)$-approximation for
  the makespan: $C_{\max} \leq \alpha(p) C_{\max}^*$. Thanks to
  Lemma~\ref{le.lower.bounds}, we know that
  $$
  M \times C_{\max}  \geq  2 n-1 = 2pk -3.
  $$
  Therefore,
  $$
  M \geq \frac{2pk-3}{C_{\max}} \geq \frac{2pk-3}{\alpha(p)(k+m-1)}.
  $$
  The approximation ratio of the studied algorithm with respect to the
  peak memory usage is, thus, bounded by:
  $$
  \beta(p) \geq \frac{M}{M^*} \geq \frac{2pk-3}{\alpha(p)  (k+m-1)(m+2)}\cdot
  $$
  Therefore, if we recall that $m=\lceil \log(p) \rceil$,
  $$
  \alpha(p)\beta(p) \geq \frac{2pk-3}{(k+m-1)(m+2)}  =
  \frac{2pk-3}{(k+\lceil\log(p)\rceil-1)(\lceil\log(p)\rceil+2)}
  \xrightarrow[k\to \infty]{}
  \frac{2p}{\lceil\log(p)\rceil+2} \cdot
  $$
  Hence, there exists a value $k_0$ such that, for any $k \geq
  k_0$,
  $$
  \alpha(p)\beta(p) \geq   \frac{2p}{\lceil\log(p)\rceil+2} -\frac{\epsilon}{2}\cdot
  $$
  This contradicts the definition of $\epsilon$ and, hence, concludes the proof.
\end{proof}

Readers may wonder whether the bound in Theorem~\ref{thm.newalphabeta}
is tight. This interrogation is especially relevant because the proof
of Theorem~\ref{thm.newalphabeta} uses the average memory usage as a
lower bound to the peak memory usage. This technique enables to design
a simple proof, which may however be very crude. In fact, in the
special case where $\alpha(p)=1$, that is, for makespan-optimal
algorithms, a stronger result holds. For that case,
Theorem~\ref{thm.newalphabeta} states that $\beta(p) \geq
\frac{2p}{\lceil\log(p)\rceil+2}$. Theorem~\ref{thm.shared.inapprox_chains}
below states that $\beta(p) \geq p-1$ (which is a stronger bound when $p
\geq 4$). This result is established through a careful, painstaking
analysis of a particular task graph. Using the average memory usage
argument on this task graph would not enable to obtain a non-trivial
bound.

\begin{theorem}\label{thm.shared.inapprox_chains}
  There does not exist any algorithm that is both optimal for makespan
  minimization and that is a $(p-1-\epsilon)$-approximation algorithm
  for the peak memory minimization, where $p$ is the number of
  processors and $\epsilon>0$.
\end{theorem}


\newcommand{\cp}[2]{\ensuremath{cp^{#1}_{#2}}}
\newcommand{\nd}[2]{\ensuremath{d^{#1}_{#2}}}
\newcommand{\lf}[3]{\ensuremath{a^{#1,#2}_{#3}}}
\newcommand{\cc}[2]{\ensuremath{c^{#1}_{#2}}}
\newcommand{\nwidgets}{\ensuremath{p-1}}

  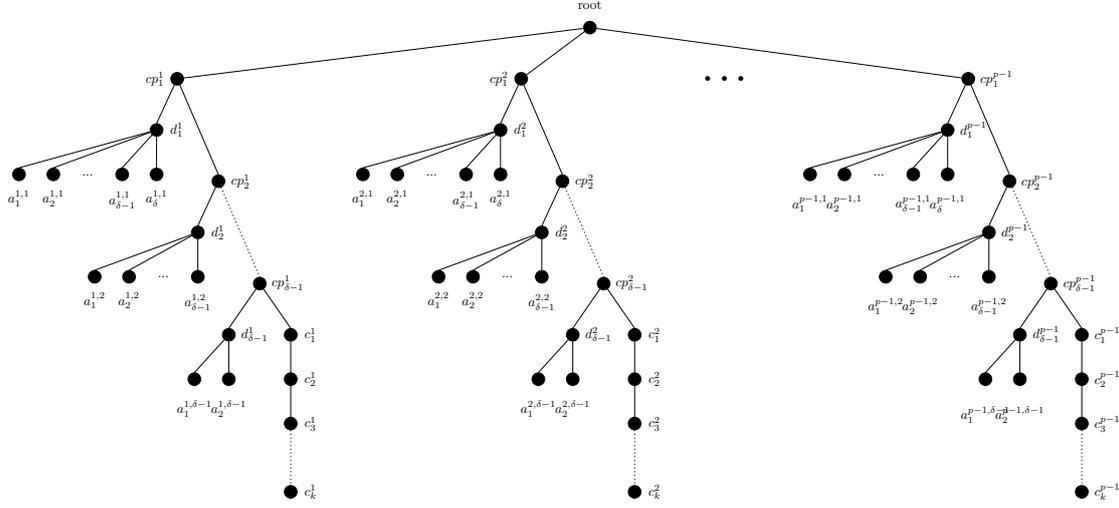
\begin{figure}
    \centering
    \resizebox{\textwidth}{!}{\begin{tikzpicture}[scale=0.9, child anchor = north, label distance = -2pt]
  \tikzstyle{every node} = [circle, fill=black, draw]
  \node (root) [label=above:root]{}
  --   (-120mm,-15mm) 
     node[label=left:\cp{1}{1}]{}
    [sibling distance=12mm, level distance=15mm]
    child{
      node[label=right:\nd{1}{1}]{} 
      [grow via three points={one child at (0,-1.3) and two children at (0,-1.3) and (-1,-1.3)}]
      child{node[label=below:\lf{1}{1}{\delta}]{}}
      child{node[label=below:\lf{1}{1}{\delta-1}]{}}
      child{node[draw=none, fill=white]{...} edge from parent [draw=none]}
      child{node[label=below:\lf{1}{1}{2}]{}}
      child{node[label=below:\lf{1}{1}{1}]{}}
    }
    child[level distance=30mm,sibling distance=24mm]{ 
      node[label=right:\cp{1}{2}]{}
      [sibling distance=12mm, level distance=15mm]
      child{
        node[label=right:\nd{1}{2}]{} 
        [grow via three points={one child at (0,-1.3) and two children at (0,-1.3) and (-1,-1.3)},sibling distance=10mm]  
        child{node[label=below:\lf{1}{2}{\delta-1}]{}}
        child{node[draw=none,fill=white]{...} edge from parent [draw=none]}
        child{node[label=below:\lf{1}{2}{2}]{}}
        child{node[label=below:\lf{1}{2}{1}]{}}
      }      
      child[level distance=30mm,sibling distance=24mm]{
        node[label=right:\cp{1}{ \delta-1}]{}    edge from parent[dotted,thick] 
        [level distance=15mm,sibling distance=18mm]
        child[solid,thin]{
          node[label=right:\nd{1}{{\delta-1}}]{} 
          [grow via three points={one child at (0,-1.3) and two children at (0,-1.3) and (-1,-1.3)},sibling distance=10mm]
          child{node[label=below:\lf{1}{\delta-1}{2}]{}}
          child{node[label=below:\lf{1}{\delta-1}{1}]{}}
        }
        child[thin,solid]{
          node[label=right:\cc{1}{1}]{} edge from parent[solid]
          [sibling distance = 10mm, level distance = 13mm]
          child{
            node[label=right:\cc{1}{2}]{}
            child{
              node[label=right:\cc{1}{3}]{}
              child[level distance=20mm]{node [label=right:\cc{1}{k}] {} edge from parent[dotted,thick]}
            }
          }
        }
      }
    }
    edge (root)
  --   (-20mm,-15mm) 
    node[label=left:\cp{2}{1}]{}
    [sibling distance=12mm, level distance=15mm]
    child{
      node[label=right:\nd{2}{1}]{} 
      [grow via three points={one child at (0,-1.3) and two children at (0,-1.3) and (-1,-1.3)}]
      child{node[label=below:\lf{2}{1}{\delta}]{}}
      child{node[label=below:\lf{2}{1}{\delta-1}]{}}
      child{node[draw=none, fill=white]{...} edge from parent [draw=none]}
      child{node[label=below:\lf{2}{1}{2}]{}}
      child{node[label=below:\lf{2}{1}{1}]{}}
    }
    child[level distance=30mm,sibling distance=24mm]{ 
      node[label=right:\cp{2}{2}]{}
      [sibling distance=12mm, level distance=15mm]
      child{
        node[label=right:\nd{2}{2}]{} 
        [grow via three points={one child at (0,-1.3) and two children at (0,-1.3) and (-1,-1.3)},sibling distance=10mm]  
        child{node[label=below:\lf{2}{2}{\delta-1}]{}}
        child{node[draw=none,fill=white]{...} edge from parent [draw=none]}
        child{node[label=below:\lf{2}{2}{2}]{}}
        child{node[label=below:\lf{2}{2}{1}]{}}
      }      
      child[level distance=30mm,sibling distance=24mm]{
        node[label=right:\cp{2}{ \delta-1}]{}    edge from parent[dotted,thick] 
        [level distance=15mm,sibling distance=18mm]
        child[solid,thin]{
          node[label=right:\nd{2}{{\delta-1}}]{} 
          [grow via three points={one child at (0,-1.3) and two children at (0,-1.3) and (-1,-1.3)},sibling distance=10mm]
          child{node[label=below:\lf{2}{\delta-1}{2}]{}}
          child{node[label=below:\lf{2}{\delta-1}{1}]{}}
        }
        child[thin,solid]{
          node[label=right:\cc{2}{1}]{} edge from parent[solid]
          [sibling distance = 10mm, level distance = 13mm]
          child{
            node[label=right:\cc{2}{2}]{}
            child{
              node[label=right:\cc{2}{3}]{}
              child[level distance=20mm]{node [label=right:\cc{2}{k}] {} edge from parent[dotted,thick]}
            }
          }
        }
      }
    }
edge (root)
  --   (40mm,-15mm) 
 node[draw=none, fill=white]{\Huge \bfseries \ldots } 
  -- (110mm,-15mm) 
    node[label=right:\cp{\nwidgets}{1}]{}
    [sibling distance=12mm, level distance=15mm]
    child{
      node[label=right:\nd{\nwidgets}{1}]{} 
      [grow via three points={one child at (0,-1.3) and two children at (0,-1.3) and (-1,-1.3)}]
      child{node[label=below:\lf{\nwidgets}{1}{\delta}]{}}
      child{node[label=below:\lf{\nwidgets}{1}{\delta-1}]{}}
      child{node[draw=none, fill=white]{...} edge from parent [draw=none]}
      child{node[label=below:\lf{\nwidgets}{1}{2}]{}}
      child{node[label=below:\lf{\nwidgets}{1}{1}]{}}
    }
    child[level distance=30mm,sibling distance=24mm]{ 
      node[label=right:\cp{\nwidgets}{2}]{}
      [sibling distance=12mm, level distance=15mm]
      child{
        node[label=right:\nd{\nwidgets}{2}]{} 
        [grow via three points={one child at (0,-1.3) and two children at (0,-1.3) and (-1,-1.3)},sibling distance=10mm]  
        child{node[label=below:\lf{\nwidgets}{2}{\delta-1}]{}}
        child{node[draw=none,fill=white]{...} edge from parent [draw=none]}
        child{node[label=below:\lf{\nwidgets}{2}{2}]{}}
        child{node[label=below:\lf{\nwidgets}{2}{1}]{}}
      }      
      child[level distance=30mm,sibling distance=24mm]{
        node[label=right:\cp{\nwidgets}{ \delta-1}]{}    edge from parent[dotted,thick] 
        [level distance=15mm,sibling distance=18mm]
        child[solid,thin]{
          node[label=right:\nd{\nwidgets}{{\delta-1}}]{} 
          [grow via three points={one child at (0,-1.3) and two children at (0,-1.3) and (-1,-1.3)},sibling distance=10mm]
          child{node[label=below:\lf{\nwidgets}{\delta-1}{2}]{}}
          child{node[label=below:\lf{\nwidgets}{\delta-1}{1}]{}}
        }
        child[thin,solid]{
          node[label=right:\cc{\nwidgets}{1}]{} edge from parent[solid]
          [sibling distance = 10mm, level distance = 13mm]
          child{
            node[label=right:\cc{\nwidgets}{2}]{}
            child{
              node[label=right:\cc{\nwidgets}{3}]{}
              child[level distance=20mm]{node [label=right:\cc{\nwidgets}{k}] {} edge from parent[dotted,thick]}
            }
          }
        }
      }
    }
edge (root);
\end{tikzpicture}}
    \caption{Tree used for establishing
      Theorem~\ref{thm.shared.inapprox_chains}.}
    \label{fig:tree_inapprox_shared_chains}
  \end{figure}

\begin{proof}
  To establish this result, we
  proceed by contradiction. Let $p$ be the number of processors.  We then assume that there exists an
  algorithm $\mathcal{A}$ which is optimal for makespan minimization
  and which is a $\beta(p)$-approximation for peak memory minimization,
  with $\beta(p) < p-1$. So, there exists $\epsilon > 0$ such that
  $\beta(p) = p - 1 -\epsilon$.

  \textbf{The tree.} Figure~\ref{fig:tree_inapprox_shared_chains}
  presents the tree used to derive a contradiction. This tree is made
  of $p-1$ identical subtrees whose roots are the children of the tree
  root. The value of $\delta$ will be fixed later on.

  \textbf{Optimal peak memory.}  A memory-optimal sequential schedule
  processes each subtree rooted at $\cp{i}{1}$ sequentially. Each of
  these subtrees can be processed with a memory of $\delta+1$ by
  processing first the subtree rooted at $\nd{i}{1}$, then the one
  rooted at $\nd {i}{2}$, etc., until the one rooted at
  $\nd {i}{\delta-1}$, and then the chain of $\cc {i}{j}$ nodes, and the
  remaining $\cp{i}{j}$ nodes. The peak memory, reached when
  processing the last subtree rooted at $\cp{i}{1}$, is thus $\delta+p-1$.

  \textbf{Optimal execution time.} The optimal execution time with $p$
  processors is at least equal to the length of the critical path. The
  critical path has a length of $\delta+k$, which is the length of the
  path from the root to any $\cc{i}{k}$ node, with $1 \leq i \leq
  p-1$. We now define $k$ for this lower bound to be an achievable
  makespan, with an overall schedule as follows:
  \begin{compactitem}
  \item Each of the first $p-1$ processors processes one of the
    critical paths from end to end (except obviously for the root node
    that will only be processed by one of them).
  \item The last processor processes all the other nodes. We define
    $k$ so that this processor finishes processing all the nodes it is
    allocated at time $k+\delta-2$. This allows, the other processors
    to process all $p-1$ nodes $\cp{1}{1}$ through $\cp{p-1}{1}$ from
    time $k+\delta-2$ to time $k+\delta-1$.
  \end{compactitem}
  In order to find such a value for $k$, we need to compute the number
  of nodes allocated to the last processor. In the subtree rooted in
  \cp{i}{1}, the last processor is in charge of processing the
  $\delta-1$ nodes \nd{i}{1} through \nd{i}{\delta-1}, and the
  descendants of the \nd{i}{j} nodes, for $1 \leq j \leq \delta-1$. As
  node \nd{i}{j} has $\delta-j+1$ descendants, the number of nodes in
  the subtree rooted in \cp{i}{1} that are allocated to the last
  processor is equal to:
  $$
  \displaystyle  (\delta-1) + \sum_{j=1}^{\delta-1} (\delta-j+1) =
  \displaystyle  \delta-2 + \frac{\delta(\delta+1)}{2} =
  \displaystyle  \frac{\delta^2+3\delta-4}{2} =
  \displaystyle  \frac{(\delta+4)(\delta-1)}{2}.
  $$
  All together, the last processor is in charge of the processing of
  $(p-1) \frac{(\delta+4)(\delta-1)}{2}$ nodes. As we have stated
  above, we want this processor to be busy from time 0 to time
  $k+\delta-2$. This gives the value of $k$:
  $$
  k+\delta-2 = (p-1)\frac{(\delta+4)(\delta-1)}{2} \quad \Leftrightarrow
  \quad k = \frac{(p-1)\delta^2+(3p-5)\delta+4(2-p)}{2}.
  $$
  Remark, by looking at the first equality, that the expression on the
  right-hand side of the second equality is always an integer;
  therefore, $k$ is well defined.

  To conclude that the optimal makespan with $p$ processors is
  $k+\delta-1$ we just need to provide an explicit schedule for the last
  processor. This processor processes all its allocated nodes, except
  node \nd{1}{1}, and nodes \lf{1}{1}{1} through \lf{1}{1}{\delta-3},
  in any order between the time 0 and $k$. Then, between time $k$ and
  $k+\delta-2$, it processes the remaining \lf{1}{1}{j} nodes and then
  node \nd{1}{1}.

  \textbf{Lower bound on the peak memory usage.} We first remark that,
  by construction, under any makespan-optimal algorithm the $p-1$
  nodes \cc{i}{j} are processed during the time interval $[k-j,
  k-j+1]$. Similarly, the $p-1$ nodes \cp{i}{j} are processed during
  the time interval $[k+\delta-j-1, k+\delta-j]$. Without loss of
  generality, we can assume that processor $P_i$, for $1 \leq i \leq
  p-1$, processes nodes \cc{i}{k} through \cc{i}{1} and then nodes
  \cp{i}{\delta-1} through \cp{i}{1} from time 0 to time
  $k+\delta-1$. The processor $P_p$ processes the other nodes. The
  only freedom an algorithm has is in the order processor $P_p$ is
  processing its allocated nodes.

  To establish a lower bound, we consider the class of schedules which
  are makespan-optimal for the studied tree, whose peak memory usage
  is minimum among makespan-optimal schedules, and which satisfy the
  additional properties we just defined. We first show that, without
  loss of generality, we can further restrict our study to schedules
  which, once one child of a node \nd{i}{j} has started being processed,
  process all the other children of that node and then the node
  \nd{i}{j} itself before processing any children of any other node
  \nd{i'}{j'}.

  We also establish this property by contradiction by assuming that
  there is no schedule (in the considered class) that satisfies the
  last property. Then, for any schedule $\mathcal{B}$, let
  $t(\mathcal{B})$ be the first date at which $\mathcal{B}$ starts the
  processing of a node \lf{i}{j}{m} while some node \lf{i'}{j'}{m'}
  has already been processed, but node \nd{i'}{j'} has not been
  processed yet. We consider a schedule $\mathcal{B}$ which maximizes
  $t(\mathcal{B})$ (note that $t$ can only take values no greater than
  $\delta+k-1$ and that the maximum and $\mathcal{B}$ are thus well
  defined). We then build from $\mathcal{B}$ another schedule
  $\mathcal{B}'$ whose peak memory is not greater and that does not
  overlap the processing of nodes \nd{i}{j} and \nd{i'}{j'}. This
  schedule is defined as follows. It is identical to $\mathcal{B}$
  except for the time slots at which node \nd{i}{j}, node \nd{i'}{j'}
  or any of their children was processed. If, under $\mathcal{B}$ node
  \nd{i}{j} was processed before node \nd{i'}{j'} (respectively, node
  \nd{i'}{j'} was processed before node \nd{i}{j}), then under the new
  schedule, in these time slots, all the children of \nd{i}{j} are
  processed first, then \nd{i}{j}, then all the children of
  \nd{i'}{j'} and finally \nd{i'}{j'} (resp. all the children of
  \nd{i'}{j'} are processed first, then \nd{i'}{j'}, then all the
  children of \nd{i}{j} and finally \nd{i}{j}). The peak memory due to
  the processing of nodes \nd{i}{j} and \nd{i'}{j'} and of their
  children is now $\max\{\delta-j+2, \delta-j'+3\}$
  (resp. $\max\{\delta-j'+2, \delta-j+3\}$). In the original schedule
  it was no smaller than $\max\{\delta-j+3, \delta-j'+3\}$ because at
  least the output of the processing of one of the children of node
  \nd{i'}{j'} (resp. \nd{i}{j}) was in memory while node \nd{i}{j}
  (resp. \nd{i'}{j'}) was processed. Hence the peak memory of the new
  schedule $\mathcal{B}'$ is no greater than the one of the original
  schedule. The new schedule satisfies the desired property at least
  until time $t+1$. Hence $t(\mathcal{B}')$ is larger than
  $t(\mathcal{B})$, which contradicts the maximality of $\mathcal{B}$
  with respect to the value $t$.

  From the above, in order to establish a lower bound on the peak memory
  of any makespan-optimal schedule, it is sufficient to consider
  only those schedules that do not overlap the processing of different
  nodes \nd{i}{j} (and of their children). Let $\mathcal{B}$ be such
  a schedule. We know that,
  under schedule $\mathcal{B}$, processor $P_p$  processes nodes without
  interruption from time 0 until time $k+\delta-2$. Therefore, in the
  time interval $[k+\delta-3,k+\delta-2]$ processor $P_p$ processes a
  node \nd{i}{1}, say \nd{1}{1}. Then we have shown that we can assume
  without loss of generality that processor $P_p$ exactly processes
  the $\delta$ children of \nd{1}{1} in the time interval
  $[k-3,k+\delta-3]$. Therefore, at time $k$, processor $P_p$ has
  processed all nodes allocated to it except node \nd{1}{1} and
  $\delta-3$ of its children. Therefore, during the time interval
  $[k,k+1]$ the memory must contain:
  \begin{compactenum}
  \item The output of the processing of all \nd{i}{j} nodes, for $1
    \leq i \leq p-1$ and $1 \leq j \leq \delta-1$, except for the
    output of node \nd{1}{1} (which has not yet been processed). This
    corresponds to $(p-1)(\delta-1) - 1$ elements.
  \item The output of the processing of 3 children of node \nd{1}{1}
    and a additional unit to store the result of a fourth one. This
    corresponds to $4$ elements.
  \item The result of the processing of the \cc{i}{1} nodes, for $ 1
    \leq i \leq p-1$ and room to store the results of the
    \cp{i}{\delta-1} nodes, for $ 1 \leq i \leq p-1$. This corresponds
    to $2(p-1)$ elements.
  \end{compactenum}
  Overall, during this time interval the memory must contain:
  $$((p-1)(\delta-1) - 1) + 4 + 2(p-1) = (p-1)\delta + p +2$$ elements.
  As the optimal peak memory is $\delta+p-1$ this gives us a lower
  bound on the ratio $\rho$ of the memory used by $\mathcal{B}$ with
  respect to the optimal peak memory usage:
  $$
  \rho \ \geq \ \frac{(p-1)\delta + p +2}{\delta+p-1}
  \   \xrightarrow[\delta\to \infty]{} \ p-1.
  $$
  Therefore, there exists a value $\delta_0$ such that
  $$
  \frac{(p-1)\delta_0 + p
    +2}{\delta_0+p-1} > p-1+\frac{1}{2}\epsilon.
  $$
  As algorithm $\mathcal{A}$ cannot have a strictly lower peak memory
  than algorithm $\mathcal{B}$ by definition of $\mathcal{B}$, this
  proves that the ratio for $\mathcal{A}$ is at least equal to
  $p-1+\frac{1}{2}\epsilon$, which contradicts the definition of
  $\epsilon$.
\end{proof}

Furthermore, a similar result can also be derived in the general model
(with arbitrary execution times and file sizes), but without the
restriction that $\alpha(p) = 1$. This is done in the next lemma.

\begin{lemma}
  \label{le.gen_alphabeta}
  When scheduling in-tree task graphs in the general model on a
  platform with $p \geq 2$ processors, there does not exist any
  algorithm that is both an $\alpha(p)$-approximation for makespan
  minimization and a $\beta(p)$-approximation for peak memory
  minimization, with $\alpha(p) \beta(p) \ < p$.
\end{lemma}

\begin{proof}
  \begin{figure*}
    \centering
    {
      \begin{tikzpicture}
        \tikzstyle{every node} = [circle, fill=black, font=\scriptsize,draw, inner sep=1.5pt, label distance = -35pt]
        \node (root) [circle,draw,label=north:{\rule[-1.5cm]{0pt}{0pt}$\begin{array}{c}n_0=0,\\ w_0=0,\\ f_0=0\end{array}$}]{}
        [sibling distance=20mm, level distance=10mm]
        child{  node[circle, draw,label=south:{\rule[1.5cm]{0pt}{0pt}$\begin{array}{c}n_1=1,\\ w_1=1,\\ f_1=0\end{array}$}]{} }
        child{  node[circle, draw,label=south:{\rule[1.5cm]{0pt}{0pt}$\begin{array}{c}n_2=1,\\ w_2=1,\\ f_2=0\end{array}$}]{} }
        child{ node[circle, draw,draw=none, fill=white]{\Huge \bfseries \ldots } edge from parent [draw=none] }
        child{  node[circle, draw,label=south:{\rule[1.5cm]{0pt}{0pt}$\begin{array}{c}n_p=1,\\ w_p=1,\\ f_p=0\end{array}$}]{} }
        ;
      \end{tikzpicture}}
    \caption{Tree used for the proof of Lemma~\ref{le.gen_alphabeta}.}
    \label{fig.tree.p.approx}
  \end{figure*}
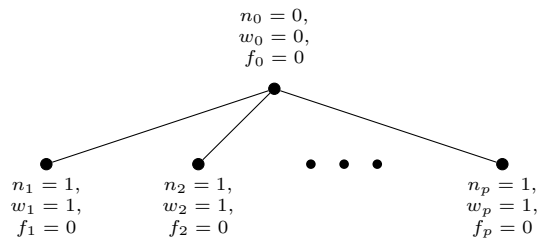
  Consider the tree drawn in Figure~\ref{fig.tree.p.approx}. This tree
  can be scheduled in time $C_{\max}^* = 1$ on $p$ processors if all
  non-root nodes are processed simultaneously (by using a peak memory
  of $p$), or sequentially in time $p$ by using only $M^* = 1$
  memory. Lemma~\ref{le.lower.bounds} states that for any schedule
  with makespan $C_{\max}$ and peak memory $M$, we have $MC_{\max}
  \geq p$. This immediately implies that any algorithm with
  approximation ratio $\alpha(p)$ for makespan and $\beta(p)$ for peak
  memory minimization must verify $\alpha(p) \beta(p) \geq p$. This bound is
  tight because, in this model, any memory-optimal sequential
  schedule is an approximation algorithm with $\beta(p) = 1$ and $\alpha(p)
  = p$.
\end{proof}

\section{Heuristics}
\label{sec.heuristics}

Given the complexity of optimizing the makespan and memory at the same
time, we have investigated heuristics and we propose six
algorithms. The intention is that the proposed algorithms cover a
range of use cases, where the optimization focus wanders between the
makespan and the required memory. The first heuristic, \ParSubtrees
(Section~\ref{sec:parallelSubtrees}), employs a memory-optimizing
sequential algorithm for each of its subtrees, the different subtrees
being processed in parallel. Hence, its focus is more on the memory
side. In contrast, \ParInnerFirst and \ParSubtrees are two list
scheduling based algorithms (Section~\ref{sec:listsched}), which
should be stronger in the makespan objective. Nevertheless, the
objective of \ParInnerFirst is to approximate a postorder in parallel,
which is good for memory in sequential. The focus of \ParDeepestFirst
is fully on the makespan. Then, we move to memory-bounded heuristics
(Section~\ref{sec:memorybounds}).  Initially, we adapt the two
list-scheduling heuristics to obtain bounds on their memory
consumption. Finally, we design a  heuristic, \BookingInnerFirst
(Section~\ref{sec:memorybooking}), that proposes a parallel execution
of a sequential postorder while satisfying a memory bound given as
input.


\subsection{Parallel execution of subtrees}
\label{sec:parallelSubtrees}


The most natural idea to process a tree $T$ in parallel is arguably
to split it into subtrees, to process each of these subtrees with
a sequentially memory-optimal
algorithm~\cite{ipdps-tree-traversal,Liu87}, and to have these
sequential processings happen in parallel. The underlying idea is to
assign to each processor a whole subtree in order to enable as much
parallelism as there are processors, while allowing to use
a single-processor memory-optimal traversal on each
subtree. Algorithm~\ref{algo.parSubtrees} outlines such an algorithm,
using Algorithm~\ref{algo.splitSubtrees} for splitting $T$ into subtrees.
The makespan obtained
using \ParSubtrees is denoted by $C_{max}^{\ParSubtreesMath}$.

\LinesNumbered
\begin{algorithm}
\DontPrintSemicolon
\caption{\ParSubtrees($T$, $p$) \label{algo.parSubtrees}}
Split tree $T$ into $q$ subtrees ($q \leq p$) and a set of remaining
nodes, using \SplitSubtrees($T$, $p$).\;
Concurrently process the $q$ subtrees, each using a memory minimizing algorithm, e.g., \cite{ipdps-tree-traversal}. \label{algo.parSubtrees.step2}\;
Sequentially process the set of remaining nodes, using a memory minimizing algorithm.\;
\end{algorithm}

In this approach, $q$ subtrees of $T$, $q \leq p$, are processed in
parallel. Each of these subtrees is a maximal subtree of $T$. In other
words, each of these subtrees includes all the descendants (in $T$) of
its root.  The nodes not belonging to the $q$ subtrees are processed
sequentially. These are the nodes where the $q$ subtrees merge, the
nodes included in subtrees that were produced in excess (if more than
$p$ subtrees were created), and the ancestors of these nodes.  An
alternative approach, as discussed below, is to process all produced
subtrees in parallel, assigning more than one subtree to each
processor when $q > p$. The advantage of
Algorithm~\ref{algo.parSubtrees} is that we can construct a splitting
into subtrees that minimizes its makespan, established shortly in
Lemma~\ref{lem:SplitSubtreesIsOptimal}.

As \len{i} is the computation weight of node $i$, \W{i} denotes the
total computation weight (i.e., sum of weights) of all nodes in the
subtree rooted in $i$, including $i$.  \SplitSubtrees uses a node
priority queue $\mathit{PQ}$ in which the nodes are sorted by
non-increasing \W{i}, and ties are broken according to non-increasing
\len{i}. $head(\mathit{PQ})$ returns the first node of $\mathit{PQ}$,
while $popHead(\mathit{PQ})$ also removes it. $\mathit{PQ}[i]$ denotes
the $i$-th element in the queue.

\SplitSubtrees starts with the root of the entire tree and continues
splitting the largest subtree (in terms of the total computation
weight \W{}) until this subtree is a leaf node ($\W{head(\mathit{PQ})}
= \len{head(\mathit{PQ})}$). The execution time of
Step~\ref{algo.parSubtrees.step2} of \ParSubtrees is that of the
largest of the $q$ subtrees of the splitting, hence
\W{head(\mathit{PQ})} for the solution found by
\SplitSubtrees. Splitting subtrees that are smaller than the largest
leaf ($\W{j}<\max_{i \in T}\len{i}$) cannot decrease the parallel
time, but only increase the sequential time. More generally, given any
splitting $s$ of $T$ into subtrees, the best execution time for $s$
with \ParSubtrees is achieved by choosing the $p$ largest subtrees for
the parallel Step~\ref{algo.parSubtrees.step2}. This can be easily
derived, as swapping a large tree included in the sequential part with
a smaller tree included in the parallel part cannot increase the total
execution time. Hence, the value $C_{max}^{\ParSubtreesMath}(s)$
computed in Step~\ref{algo.splitSubtrees.step10} is the makespan that
would be obtained by \ParSubtrees on the splitting computed so far. At
the end of algorithm \SplitSubtrees
(Step~\ref{algo.splitSubtrees.step12}), the splitting which yields the
smallest makespan is selected.


\begin{algorithm}
  \DontPrintSemicolon
  \caption{\SplitSubtrees($T$, $p$)\label{algo.splitSubtrees}}
  \lForEach{node $i$}{compute $\W{i}$ (the total processing time of
    the tree rooted at $i$)}\;
  Initialize priority queue $\textit{PQ}$ with the tree root\;
  $\textit{seqSet} \leftarrow \emptyset$\;
  $Cost(0)=\W{root}$\;
  $s \leftarrow 1$ \tcc*{splitting rank}
  \While {$\W{head(\mathit{PQ})} > \len{head(\mathit{PQ})}$} {
    $\textit{node} \leftarrow popHead(\mathit{PQ})$\;
    $\textit{seqSet} \leftarrow \textit{seqSet} \cup \textit{node}$\;
    Insert all children of $\textit{node}$ into priority queue $\textit{PQ}$\;
    $C_{max}^{\ParSubtreesMath}(s) = \W{head(\mathit{PQ})} +
      \sum_{i \in seqSet}\len{i} + \sum_{i=\mathit{PQ}[p+1]}^{|\mathit{PQ}|}\W{i}$\label{algo.splitSubtrees.step10}\;
    $s \leftarrow s+1$\;
  }
  Select splitting $S$ with $C_{max}^{\ParSubtreesMath}(S)=\min_{t=0}^{s-1} C_{max}^{\ParSubtreesMath}(t)$ (break ties in favor of smaller $t$)\label{algo.splitSubtrees.step12}\;
\end{algorithm}

\begin{lemma}\label{lem:SplitSubtreesIsOptimal}
 \SplitSubtrees returns a splitting of $T$ into subtrees that results in the \textit{makespan}-optimal processing of $T$ with \ParSubtrees.
\end{lemma}

\begin{proof}
 The proof is by contradiction. Let $S$ be the splitting into subtrees
 selected by \SplitSubtrees. Assume now that there is a different
 splitting $S_{opt}$ which results in a strictly shorter processing with \ParSubtrees.

 Because of the termination condition of the \emph{while}-loop,
 \SplitSubtrees splits any subtree that is heavier than the heaviest
 leaf. Therefore, any such tree will be at one time at the head of the
 priority queue. Let $r$ be the root node of a heaviest subtree in
 $S_{opt}$. From what precedes, there always exists a step $t$ which
 is the first step in \SplitSubtrees where a node, say $r_t$, of
 weight \W{r}, is the head of $\mathit{PQ}$ at the end of the step
 ($r_t$ is not necessarily equal to $r$, as there can be more than one
 subtree of weight \W{r}). 
 Let $S_t$ be the solution built by \SplitSubtrees at the end of step
 $t$.  By definition of $r$, there cannot be any leaf node in the entire
 tree that is heavier than \W{r}. The cost of the solution $S_t$ is equal
 to the execution time of the parallel processing of the $\min\{q,
 p\}$ subtrees plus the execution time of the sequential processing of
 the remaining nodes.  Therefore
 $C_{max}^{\ParSubtreesMath}(t)=\W{r}+Seq(t)$, where $Seq(t)$ is the
 total weight of the sequential set $seqSet(t)$ plus the total weight
 of the surplus subtrees (that is, of all the subtrees in
 $\mathit{PQ}$ except the $p$ subtrees of largest weights). The cost of
 $S_{opt}$ is $C_{max}^{*}=\W{r}+Seq(S_{opt})$, given that $r$ is the
 root of a heaviest subtree of $S_{opt}$ by definition.

 \SplitSubtrees splits subtrees by non-increasing
 weights. Furthermore, by definition of step $t$, all subtrees split
 by \SplitSubtrees, up to step $t$ included, were subtrees whose
 weights were strictly greater than \W{r}. Therefore, because $r$ is
 the weight of the heaviest subtree in $S_{opt}$, all the subtrees
 split by \SplitSubtrees up to step $t$ included must have been split
 to obtain the solution $S_{opt}$.  This has several
 consequences. Firstly, $seqSet(t)$ is a subset of $seqSet(S_{opt})$,
 because, for any solution $S$, $seqSet(S)$ is the set of all nodes
 that are roots of subtrees split to obtain the solution
 $S$. Secondly, either a subtree of $S_t$ belongs to $S_{opt}$ or this
 subtree has been split to obtain $S_{opt}$. Therefore, the sequential
 processing of the $\max\{q-p, 0\}$ exceeding subtrees is no smaller
 in $S_{opt}$ than in the solution built at step $t$. It directly
 follows from the two above consequences that $Seq(t) \leq
 Seq(S_{opt})$.  However, $S_{opt}$ and $S_t$ have the same execution
 time for the parallel phase \W{r}. It follows that
 $C_{max}^{\ParSubtreesMath}(t) \leq C_{max}^{*}$, which is a
 contradiction to $S_{opt}$'s shorter processing time.
\end{proof}

\paragraph{Complexity} We first analyze the complexity of
\SplitSubtrees.  Computing the weights \W{i} costs $O(n)$.  Each
insertion into $\mathit{PQ}$ costs $O(\log{n})$ and calculating
$C_{max}^{\ParSubtreesMath}(s)$ in each step costs
$O(p)$. 
Given that there are $O(n)$ steps, \SplitSubtrees's complexity is
$O(n(\log{n}+p))$.  The complexity of the sequential traversal
algorithms used in Steps 2 and 3 of \ParSubtrees is at most
$O(n^2)$, e.g., \cite{ipdps-tree-traversal,Liu87}, or $O(n\log{n})$ if
the optimal postorder is sufficient. Thus the total complexity
of \ParSubtrees is $O(n^2)$ or $O(n(\log{n}+p))$, depending on the
chosen sequential algorithm.


\paragraph{Memory}

\begin{lemma}\label{parsubtrees.memory}
  \ParSubtrees is a $p$-approximation algorithm for peak memory
  minimization: the peak memory, $M$, verifies $M \leq p
  M_{seq}$, where $M_{seq}$ is the memory required for the complete
  sequential execution.
\end{lemma}

\begin{proof}
  We first remark that during the parallel part of \ParSubtrees, the
  total memory used, $M_p$, is not more than $p$ times $M_{seq}$.  Indeed, each of
  the $p$ processors executes a maximal subtree and the
  processing of any subtree does not use, obviously, more memory (if done
  optimally) than the processing of the whole tree.  Thus, $M_{p} \leq
  p \cdot M_{seq}$.

  During the sequential part of \ParSubtrees, the memory used, $M_S$, is
  bounded by $M_{seq} + \sum_{i \in Q}\f{i}$, where the second term is
  for the output files produced by the root nodes of the $q \leq p$ subtrees processed
  in parallel ($Q$ is the set of the root nodes of the $q$ trees processed
  in parallel). We now claim that at least two of those subtrees have a
  common parent. More specifically, let us denote by $X$ the node that
  was split last ({\it i.e.}, it was split in the step $S$ which is
  selected at the end of \SplitSubtrees). Our claim is that at least
  two children of $X$ are processed in the parallel part. Before $X$
  was split (in step $S-1$), the makespan as computed in Step
  \ref{algo.splitSubtrees.step10} of \SplitSubtrees is $C_{max}(S-1) =
  W_X + Seq(S-1)$, where $Seq(S-1)$ is the work computed in
  sequential ($\sum_{i \in seqSet}\len{i} +
  \sum_{i=\mathit{PQ}[p+1]}^{|\mathit{PQ}|}\W{i}$). Let
  $D$ denote the set of children of $X$ which are not executed in parallel,
  then the total weight of their subtrees is $W_D = \sum_{i \in D}
  W_i$. We now show that if at most one child of $X$ is processed in the
  parallel part, $X$ was not the node that was split last:
\begin{itemize}
\item If exactly one child $C$ of $X$ is processed in the parallel
  part, then $C_{max}(S) = W_{X'} + Seq(S-1) + w_X + W_D$, where
  $X'$ is the new head of the queue, and thus verifies $W_{X'} \geq
  W_C$. And since $W_X = w_X + W_C + W_D$, we can conclude that
  $C_{max}(S) \geq C_{max}(S-1)$.
\item If no child of $X$ is processed in the parallel part, then
  $C_{max}(S) = W_{X'} + Seq(S-1) - W_{Y} + w_X + W_D$, where $X'$
  is the new head of the queue and $Y$ is the newly inserted node in
  the $p$ largest subtrees in the queue. Since $W_{X'} \geq W_{Y}$ and
  $W_X = w_X + W_D$, we obtain once again $C_{max}(S) \geq
  C_{max}(S-1)$.
\end{itemize}
In both cases we have $C_{max}(S) \geq C_{max}(S-1)$, which contradicts
the definition of $X$ (the select phase, Step~\ref{algo.splitSubtrees.step12} of \SplitSubtrees, would have selected
step $S-1$ rather than step $S$). Let us now denote by $C_1$ and $C_2$
two children of $X$ which are processed in the parallel
phase. Remember that the memory used during the sequential part is
bounded by $M_S\leq M_{seq} + f_{C_1} + f_{C_2} + \sum_{i \in
  Q\setminus\{C_1, C_2\}}\f{i}$.  Since a sequential execution must
process node $X$, we obtain $f_{C_1} + f_{C_2} \leq M_{seq}$. And
since $\forall i, f_{i} \leq M_{seq}$, we can bound the memory used
during the sequential part by $M_S \leq 2M_{seq} + (p-2)M_{seq} \leq
pM_{seq}$.

Furthermore, given that up to $p$ processors work in parallel, each on
its own subtree,
it is easy to see that this bound is tight if the sequential peak memory can be reached
in each subtree.
\end{proof}

\paragraph{Makespan} \ParSubtrees delivers a $p$-approximation
algorithm for makespan minimization, and this bound is tight. Because
at least one processor is working at any time
under \ParSubtrees, \ParSubtrees delivers, in the worst case, a
$p$-approximation for makespan minimization.  To prove that this bound
is tight, we consider a tree of height 1 with $p \cdot k$ leaves (a
fork), where all execution times are equal to 1 ($\forall i \in
T$, $\len{i}=1$), and where $k$ is a large integer (this tree is
depicted in Figure~\ref{fig:fork}).  The optimal makespan for such a
tree is $C_{max}^{*}=kp/p+1=k+1$ (the leaves are processed in
parallels, in batches of size $p$, and then the root is processed).
With \ParSubtrees $p$ leaves are processed in parallel, and then
the remaining nodes are processed sequentially. The makespan is thus
$C_{max}=(1+pk-p)+1=p(k-1)+2$. When $k$ tends to $+\infty$ the
ratio between the makespans tends to $p$.

\begin{figure}
  \centering
  \scalebox{0.8}{
\begin{tikzpicture}[scale=0.9, child anchor = north]
\tikzstyle{every node}=[ellipse, draw]
  \node{root}
  [sibling distance=25mm, level distance=15mm]
  child{node{1}}
  child{node{2}}
  child{node[draw=none]{...} edge from parent [draw=none]}
  child{node{$p \cdot k-1$}}
  child{node{$p \cdot k$}}
  ;
\end{tikzpicture}
}
  \caption{\ParSubtrees is at best a $p$-approximation for the makespan.}
  \label{fig:fork}
\end{figure}
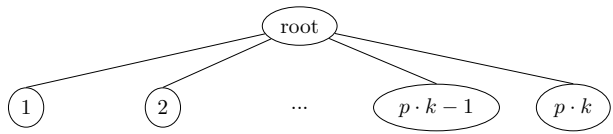

Given the just observed worst case for the makespan, a makespan
optimization for \ParSubtrees is to allocate all produced subtrees to
the $p$ processors instead of only $p$ subtrees. This can be done by
ordering the subtrees by non-increasing total weight and allocating
each subtree in turn to the processor with the lowest total
weight. Each of the parallel processors executes its subtrees
sequentially. This optimized form of the algorithm is
named \ParSubtreesOptim. Note that this optimization should improve
the makespan, but it will likely worsen the peak memory usage.


\subsection{List scheduling algorithms}
\label{sec:listsched}

\ParSubtrees is a high-level algorithm employing sequential
memory-optimized algorithms. An alternative, explored in this section,
is to design algorithms that directly work on the tree in parallel.
We first present two such algorithms that are event-based list
scheduling algorithms~\cite{Hwang1989:spg}.  One of the strong points
of list scheduling algorithms is that they are
$(2-\frac{1}{p})$-approximation algorithms for makespan
minimization~\cite{GrahamList}.

Algorithm~\ref{algo.listScheduling} outlines a generic list
scheduling, driven by node finish time events. At each event at least
one node has finished so at least one processor is available for
processing nodes. Each available processor is given the respective
head node of the priority queue. The priority of nodes is given by the
total order $O$, a parameter to Algorithm~\ref{algo.listScheduling}.

\begin{algorithm}
\DontPrintSemicolon
\caption{List scheduling($T$, $p$, $O$) \label{algo.listScheduling}}
  Insert leaves in priority queue $\mathit{PQ}$ according to order $O$\;
  $\mathit{eventSet} \leftarrow \{0\}$   \tcc*{ascending order}
  \While(\tcc*[f]{event:node finishes}){$\mathit{eventSet} \neq \emptyset$}{
    $popHead(\mathit{eventSet})$\;
    $\Done \leftarrow$ set of the new ready nodes\;
    Insert nodes from \Done in $\mathit{PQ}$ according to order $O$\tcc*{available parents of nodes completed at event}
    $\mathcal{P} \leftarrow$ available processors\;
    \While{$\mathcal{P} \neq \emptyset$ and $\mathit{PQ} \neq \emptyset$}{
      $\mathit{proc} \leftarrow \mathit{popHead}(\mathcal{P})$\;
      $\mathit{node} \leftarrow \mathit{popHead}(\mathit{PQ})$\;
      Assign $\mathit{node}$ to $\mathit{proc}$\;
      $\mathit{eventSet}~\leftarrow~\mathit{eventSet} \cup \mathit{finishTime}(\mathit{node})$\;
    }
   }
\end{algorithm}

\subsubsection{Heuristic \ParInnerFirst}

From the study of the sequential case, one knows that a
\textit{postorder} traversal, while not optimal for all instances,
provides good results~\cite{ipdps-tree-traversal}.  Our intention is
to extend the principle of postorder traversal to the parallel
processing. For the first heuristic, called \ParInnerFirst, the
priority queue uses the following ordering $O$: 1) inner nodes, in an
arbitrary order; 2) leaf nodes ordered according to a given postorder
traversal. It makes heuristic sense that the chosen postorder is an
optimal sequential postorder, so that memory consumption can be
minimized. We do not further define the order of inner nodes
because it has absolutely no impact. Indeed, because we target the
processing of tree-shaped task-graphs, the processing of a node makes
at most one new inner node available, and the processing of this new
inner node can start right away on the processor that freed it by
completing the processing of its last un-processed child.



\paragraph{Complexity}
The complexity of \ParInnerFirst is that of determining the input
order $O$ and that of the list scheduling. Computing the optimal
sequential postorder is $O(n\log{n})$~\cite{Liu86}. In the list
scheduling algorithm there are $O(n)$ events and $n$ nodes are
inserted and retrieved from $\mathit{PQ}$. An insertion into
$\mathit{PQ}$ costs $O(\log{n})$, so the list scheduling complexity is
$O(n\log{n})$. Hence, the total complexity is also $O(n\log{n})$.

\paragraph{Memory}
\ParInnerFirst is not an approximation algorithm with respect to peak
memory usage. This is derived considering the tree in
Figure~\ref{fig:NoMemBoundInnerFirst}. All output files have size 1
and the execution files have size 0 ($\forall i \in T: \f{i}=1,\n{i}=0$).
Under an optimal sequential processing, leaves are
processed in a deepest first order. The resulting optimal memory
requirement is $M_{seq}=p+1$, reached when processing a join node.
With $p$ processors, all leaves have been processed at the time the
first join node ($k-1$) can be executed. (The longest chain has length
$2k-2$.) At that time there are $(k-1)\cdot(p-1)+1$ files in
memory. When $k$ tends to $+\infty$ the ratio between the memory
requirements also tends to $+\infty$.

\begin{figure}
  \centering
    \begin{tikzpicture}[scale=0.7,
      level/.style={level distance=12mm, sibling distance=7mm},
      vertex/.style={circle,solid,draw},
      ghost/.style={circle}]

      \makeatletter
      \let\mypgfutil@ifnch\pgfutil@ifnch
      \def\pgfutil@ifnch{%
        \let\x@next\@empty
        \ifx\pgfutil@let@token\Chain\let\pgfutil@let@token C\let\x@next\expandafter\fi
        \ifx\pgfutil@let@token\children\let\pgfutil@let@token c\let\x@next\expandafter\fi
        \ifx\pgfutil@let@token\mynolabelchildren\let\pgfutil@let@token m\let\x@next\expandafter\fi
      \x@next\mypgfutil@ifnch}
    \makeatother

      \newcommand{\Chain}{%
        node[vertex,label=right:$k$] {} edge from parent[solid] child{
          node[vertex] {} edge from parent[solid] child{node[vertex]
            {} edge from parent[dotted]
            child{node[vertex,label=right:$2k$] {} edge from
              parent[solid]}}}%
      }%

      \newcommand{\children}{%
        child{node[vertex,label=below:1] {} edge from parent[solid]}
        child{node[ghost] {$\cdots$} edge from parent[draw=none]}
	child{node[vertex] {} edge from parent[solid]}
        child{node[ghost] {$\cdots$} edge from parent[draw=none]}
        child{node[vertex,label=below:$p-1$] {} edge from parent[solid]}
        child{node[ghost] {} edge from parent[draw=none]}%
      }%

      \newcommand{\mynolabelchildren}{%
        child{node[vertex] {} edge from parent[solid]}
        child{node[ghost] {$\cdots$} edge from parent[draw=none]}
	child{node[vertex] {} edge from parent[solid]}
        child{node[ghost] {$\cdots$} edge from parent[draw=none]}
        child{node[vertex] {} edge from parent[solid]}
        child{node[ghost] {} edge from parent[draw=none]}%
      }%
     \footnotesize
     \node[vertex,label=right:1] (root) {}
       \children
       child{node[vertex] {}
         child{node[vertex,right=2] {} edge from parent[dotted]
	   \mynolabelchildren
	   child{node[vertex,label=right:$k-1$] {} edge from parent[solid]
             \children
             child{node[vertex,label=right:$k$] {} edge from parent[solid] child{ node[vertex] {} edge from parent[solid] child{node[vertex] {} edge from parent[dotted] child{node[vertex,label=right:$2k-2$] {} edge from parent[solid]}}}}
           }
         }
       }
     ;
   \end{tikzpicture}
      \caption{No memory bound for \ParInnerFirst.}\label{fig:NoMemBoundInnerFirst}
    \end{figure}
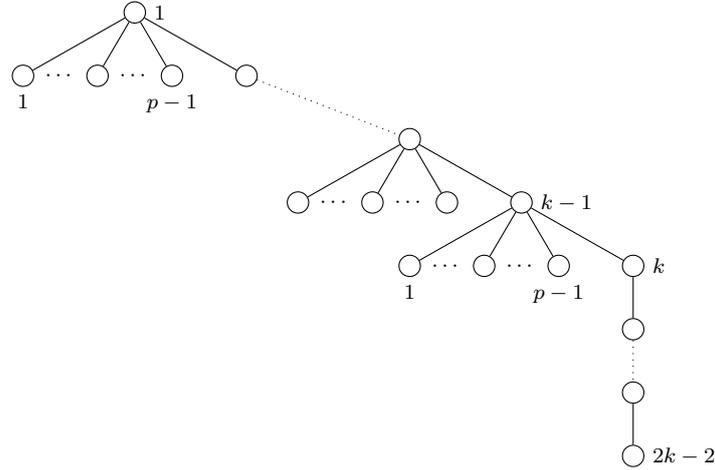

%

\subsubsection{Heuristic \ParDeepestFirst}
\label{sec.parDeepestFirst}
The previous heuristic, \ParInnerFirst, tries to take advantage of the
memory performance of optimal sequential postorders. Going in the
opposite direction, another heuristic objective can be the
minimization of the makespan. For trees, an inner node depends on all
the nodes in the subtree it defines. Therefore, it makes heuristic
sense to try to process the deepest nodes first to try to reduce any
possible waiting time. For the parallel processing of a tree, the most
meaningful definition of the depth of a node $i$ is the
\len{}-weighted length of the path from $i$ to the root of the tree,
including \len{i} (therefore, the depth of node $i$ is equal to its
top-level plus \len{i}~\cite{livre6}). A deepest node in the tree is
a deepest node in a critical path of the tree.

\ParDeepestFirst is our proposed list-scheduling deepest-first
heuristic. \ParDeepestFirst is defined by
Algorithm~\ref{algo.listScheduling} called with the following node
ordering $O$: nodes are ordered according to their depths and, in case
of ties, inner nodes have priority over leaf nodes, and remaining ties
are broken according to an optimal sequential postorder.



\paragraph{Complexity}
The complexity is the same as for \ParInnerFirst, namely $O(n\log{n})$. See \ParInnerFirst's complexity analysis.


\paragraph{Memory}
The memory required by \ParDeepestFirst is unbounded with respect to
the optimal sequential memory $M_{seq}$. Consider the tree in
Figure~\ref{fig.combWithChains} with many long chains, assuming the
Pebble Game model (i.e., $\forall i \in T : \f{i}=1, \n{i}=0, \len{i}=1$).
The optimal sequential memory requirement is
3. 
The memory usage of \ParDeepestFirst will be proportional to
the number of leaves, because they are all at the same depth, the
deepest one. As we can build a tree like the one of
Figure~\ref{fig.combWithChains} for any predefined number of chains,
the ratio between the memory required by \ParDeepestFirst and the
optimal one is unbounded.

    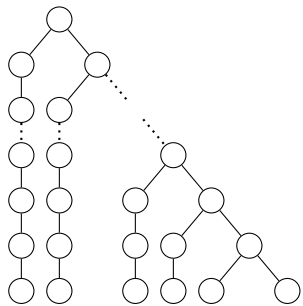
\begin{figure}
    \centering
    \begin{tikzpicture}[scale=0.5]
     \tikzstyle{level 1}=[level distance=12mm, sibling distance=20mm]
     \tikzstyle{vertex}=[circle,draw] 

     \node[vertex]{}
       child{node[vertex]{}
         child{node[vertex]{}
           child{node[vertex]{} edge from parent[dotted,thick]
             child[solid,thin]{node[vertex]{}
               child{node[vertex]{}
                 child{node[vertex]{}}
               }
             }
           }
         }
       }
       child{node[vertex]{}
         child{node[vertex]{}
           child{node[vertex]{} edge from parent[dotted,thick]
             child[solid,thin]{node[vertex]{}
               child{node[vertex]{}
                 child{node[vertex]{}}
               }
             }
           }
         }
         child{node[draw=none]{} edge from parent[dotted,thick]
           child[]{node[draw=none]{} edge from parent[draw=none]
             child{node[draw=none]{} edge from parent[draw=none]
               child{node[draw=none]{} edge from parent[draw=none]
                 child{node[draw=none]{} edge from parent[draw=none]}
               }
             }
           }
           child[thin,solid]{node[vertex]{} edge from parent[dotted, thick]
             child[thin,solid]{node[vertex]{}
               child{node[vertex]{}
                 child{node[vertex]{}}
               }
             }
             child[thin,solid]{node[vertex]{}
               child{node[vertex]{}
                 child{node[vertex]{}}
               }
               child{node[vertex]{}
                 child{node[vertex]{}}
                 child{node[vertex]{}}
               }
             }
           }
         }
       }
     ;

    \end{tikzpicture}
      \caption{Tree with long chains.} \label{fig.combWithChains}
    \end{figure}

\subsection{Memory bounded heuristics}
\label{sec:memorybounds}

From the analysis of the three algorithms presented so far we have
seen that only \ParSubtrees gives a guaranteed bound on the required
peak memory. The memory behavior of the two other
algorithms, \ParInnerFirst and \ParDeepestFirst, will be analyzed in
the experimental evaluation presented in
Section~\ref{sec.experiments}. In a practical setting it might be very
desirable to have a strictly bounded memory consumption, so as to be
certain that the algorithm can be executed with the available
memory. In fact, a guaranteed upper limit might be more important than
a good average behavior as the system needs to be equipped with
sufficient memory for the worst case. \ParSubtrees's guarantee of at
most $p$ times the optimal sequential memory seems high,
and thus an obvious goal would be to have a heuristic that minimizes
the makespan while keeping the peak memory usage below a given
bound. In order to approach this goal, we first study how to limit the
memory consumption of \ParInnerFirst and \ParDeepestFirst. Our study
relies on some reduction property on trees and we show how to
transform general trees into reduction trees in Section~\ref{sec:reductiontrees}.
We then develop memory bounded versions of \ParInnerFirst
and \ParDeepestFirst (Section~\ref{sec:boundedlistsched}).  The memory
bounds achieved by these new variants are rather lax. Therefore, we
design our last heuristic \BookingInnerFirst, with stronger
memory properties (Section~\ref{sec:memorybooking}). In the
experimental section (Section~\ref{sec.experiments}), we will show
that these three heuristics achieve different trade-offs between
makespan and memory usage.

\subsubsection{Simplifying tree properties}
\label{sec:reductiontrees}

To design our memory-bounded heuristics, we make two simplifying
assumptions. First, the considered trees do not have any execution
files. In other words, we assume that, for any task $i$, $n_i = 0$.

\paragraph{Eliminating execution files} To still be able to deal with
general trees, we can transform any tree $T$ with execution files into
a strictly equivalent tree $T'$ where all execution files have a null size. Let $i$
be any node of $T$. We add to $i$ a new leaf child $i'$ whose execution time
is null ($\len{i'}=0$), whose execution file is of null size ($\n{i'}=0$),
and whose output file has size $\n{i}$ ($\f{i'}=\n{i}$).
Then we set $\n{i}$ to 0. Any schedule $S$ for the original tree $T$
can be easily transformed into a schedule $S'$ for the new tree $T'$
with the exact same memory and execution-time characteristics: $S'$
schedules a node from $T$ at the same time than $S$, and a node $i$ from
$T'\setminus T$ at the same time than the father of $i$ is scheduled by $S$
(because $i$ has a null execution time).
\\

\noindent The second simplifying assumption is that all considered trees
are reduction trees:

\begin{definition}[reduction tree]
  A task tree is a \textbf{reduction tree} if the size of the output
  file of any inner node $i$ is not more than the sum of its input files:
  \begin{equation}
    \label{eq.f-redution}
    \f{i} \leq \sum_{j \in \Children{i}} \f{j}.
  \end{equation}
\end{definition}

This reduction property is very useful, because it implies that
executing an inner node does not increase the amount of memory needed
(this will be used for instance in Theorem~\ref{theo.memLimitedAlgos}).

For convenience, we sometimes use the following notation to denote
the sum of the sizes of the input files of an inner node $i$:
$$
\inputs{i} = \sum_{j\in \Children{i}} f_j.
$$

We now show how general trees can be transformed into reduction trees.

\paragraph{Turning trees into reduction trees} We can transform any
tree $T$ that does not satisfy the reduction property stated by
Equation~\eqref{eq.f-redution} into a tree where each (inner) node
satisfies it. Let $i$ be any inner node of $T$. We add to $i$ a new leaf
child $i'$ whose execution time is null ($\len{i'}=0$), whose
execution file is of null size ($\n{i'}=0$), and whose output file has
size: $$
\f{i'}=\max\left\{0,\f{i}-\left(\sum_{j \in \Children{i}}
    \f{j}\right)\right\}=\max\{0, \f{i}-\inputs{i}\}.
    $$

The new tree is not equivalent to the original one. Let us consider an
inner node $i$ that did not satisfy the reduction property. Then,
$\f{i'}= \f{i}-\inputs{i} >0$. The memory used to execute node $i$ in
the tree $T$ is: $\inputs{i}+\n{i}+\f{i}$. In the new tree, the memory
needed to execute this node is: $(\inputs{i}+(\f{i}-\inputs{i}) +
\n{i}+\f{i} > \inputs{i}+\n{i}+\f{i}$. Any schedule of the original
tree can be transformed into a schedule of the new tree with an
increase of the memory usage bounded by:
$$ p \times \max_i \{0, \f{i}-\inputs{i}\}.$$ Obviously, a more
clever approach is to transform a tree first into a tree without
execution files, and then to transform the new tree into a tree with the
reduction property. Under this approach, the increase of the memory
usage is bounded by:
$$
p \times \max_i \{0, \f{i}-\inputs{i}-\n{i}\}.$$

\paragraph{Transforming schedules} The algorithms proposed in the following
subsections produce schedules for reduction trees without execution files, which
might have been created from general trees which do not possess our simplifying
properties. The schedule $S'$ produced by an algorithm for a reduction tree
without execution files $T'$ can be readily transformed into a schedule $S$ for
the original tree $T$. To create schedule $S$, we simply remove all (leaf) nodes
from the schedule $S'$ that were introduced in the simplification transform
($i' \in T' \setminus T$). Because those nodes have zero processing time ($\forall i' \in
T' \setminus T: \len{i'}=0$) there is no impact on the ordering and on
the starting time of the
other nodes of $T$. In terms of memory consumption, the peak memory for schedule $S$
is never higher than that for schedule $S'$. A leaf $i'$ that was added to eliminate an
execution file might use memory earlier in $S'$ than the execution file $\n{i}$ in $S$,
but it is the same amount and freed at the same time. In terms of leaf nodes
introduced to enforce the reduction property, they might only increase the memory
needed for tree $T'$ (as discussed above); hence, removing these nodes can not
increase the peak memory needed for schedule $S$. In summary, the schedule $S$ for tree
$T$ has the same makespan as $S'$ and a peak memory that is not
greater than that of $S'$.

\subsubsection{Memory-bounded \ParInnerFirst and \ParDeepestFirst}
\label{sec:boundedlistsched}

Both \ParInnerFirst and \ParDeepestFirst are based on the list
scheduling approach presented in
Algorithm~\ref{algo.listScheduling}. To achieve a memory bounded
version of these algorithms for reduction trees, we modify
Algorithm~\ref{algo.listScheduling} to obtain
Algorithm~\ref{algo.listSchedulingMemLimit}. The code common to both
algorithms is shown in gray in
Algorithm~\ref{algo.listSchedulingMemLimit} and the new code is
printed in black.

We use the same event concept as previously. However, we only start
processing a node if i) it is an inner node; or ii) it is a leaf node
and the current memory consumption plus the leaf's output file
($\f{c}$) is less than the amount $M$ of available memory. Once a node
is encountered that fulfills neither of these conditions, the node
assignment is stopped ($\mathcal{P} \leftarrow \emptyset$) until the next
event. Therefore, Algorithm~\ref{algo.listSchedulingMemLimit} may
deliberately keep some processors idle when there are available tasks,
and thus does not necessarily produce a list schedule (hence, the name
of ``pseudo'' list schedules). Subsequently, the only approximation
guarantee on the makespan produced by these heuristics is that they
are $p$-approximations, the worst case for heuristics that always use
at least one processor at any time before the entire processing
completes.

The algorithms based on this memory-bounded scheduling approach are
called \ParInnerFirstMemLimit and \ParDeepestFirstMemLimit.


\begin{algorithm}
  \DontPrintSemicolon
  \caption{Pseudo list scheduling with memory limit ($T$, $p$, $O$, $M$) \label{algo.listSchedulingMemLimit}}
  \textcolor{mygray} {Insert leaves in priority queue $\mathit{PQ}$ according to order $O$\;
    \textcolor{mygray} {$\mathit{eventSet} \leftarrow \{0\}$   \tcc*{ascending order}}
    \textcolor{black} {$M_{\mathit{used}} \leftarrow 0$ \tcc*{amount of memory used}}
    \While(\tcc*[f]{event:node finishes}){$\mathit{eventSet} \neq \emptyset$}{
      $popHead(\mathit{eventSet})$\;
      $\Done \leftarrow$ set of the new ready nodes\;
      Insert nodes from \Done in $\mathit{PQ}$ according to order $O$\tcc*{available parents of nodes completed at event}
      $\mathcal{P} \leftarrow$ available processors\;
      \textcolor{black}{$M_{\mathit{used}} \leftarrow M_{\mathit{used}} - \sum_{j \in \Done} \inputs{j}$}\;
      \While{$\mathcal{P} \neq \emptyset$ and $\mathit{PQ} \neq \emptyset$}{
        \textcolor{black} {$c \leftarrow head(\mathit{PQ})$          \label{algo.line.candidate}\;
          \If{$|\Children{c}| > 0$ \textbf{or} $M_{\mathit{used}} + \f{c} \leq M$}{
          \label{algo.line.check}
            $M_{\mathit{used}} \leftarrow M_{\mathit{used}} + \f{c}$\;
            \textcolor{mygray} {$proc \leftarrow popHead(\mathcal{P})$}\;
            \textcolor{mygray} {$node \leftarrow popHead(\mathit{PQ})$}\;
            \textcolor{mygray} {Assign $node$ to $proc$}\;
            \textcolor{mygray} {$\mathit{eventSet} \leftarrow \mathit{eventSet} \cup finishTime(node)$}\;
          }
          \Else{$\mathcal{P} \leftarrow \emptyset$}
        }
      }
    }
  }
\end{algorithm}


\begin{theorem}
  \label{theo.memLimitedAlgos}
  The peak memory requirement of \ParInnerFirstMemLimit
  and \ParDeepestFirstMemLimit for a reduction tree without execution files
  processed with a memory bound $M$ and a node order $O$ is at most $2M$,
  if $M \geq M_{seq}$, where $M_{seq}$ is the peak memory usage of the
  corresponding sequential algorithm with the same node order $O$.
\end{theorem}

\begin{proof}
  We first show that the required memory never exceeds $2M$ and then
  we show that the algorithms completely process the considered tree
  $T$.

  We analyze the memory usage at the time a new candidate node $c$ is
  considered for execution (line~\ref{algo.line.candidate} of
  Algorithm~\ref{algo.listSchedulingMemLimit}). The amount of
  currently used memory is then $M_{\mathit{used}} = \mathit{In}_\textsc{in} +
  \mathit{Out}_\textsc{in} + \mathit{Out}_\textsc{lf}
  + \mathit{InIdle}$, where:
  \begin{compactitem}
  \item $\mathit{In}_\textsc{in}$ is the size of the input files of
    the currently processed inner nodes;
  \item $\mathit{Out}_\textsc{in}$ is the size of the output files of
    the currently processed inner nodes;
  \item $\mathit{Out}_\textsc{lf}$ is the size of the output files of
    the currently processed leaves;
  \item $\mathit{InIdle}$ is the size of the input files stored in
    memory but not currently used (because they are input files of
    inner nodes that are not yet ready).
  \end{compactitem}

  There are two cases, the candidate node $c$ can be either a leaf node or an inner node:
  \begin{compactenum}
  \item\label{thm6.case1} $c$ is a leaf node.  The processing of a leaf node only starts
    if $M_{\mathit{used}} + \f{c} \leq M$. Therefore, the processing of a leaf
    node never provokes the violation of the memory bound of $M$ and,
    thus, a fortiori, of a memory limit of $2M$.

  \item\label{thm6.case2} $c$ is an inner node. The processing of a
    candidate inner node always starts right away, regardless of the
    amount of available memory. When the processing of $c$ starts, the
    amount of required memory becomes $M_{\mathit{new}} =
    \mathit{In}_\textsc{in} + \mathit{Out}_\textsc{in} +
    \mathit{Out}_\textsc{lf} + \mathit{InIdle} + \f{c}$. $T$ is by
    hypothesis a reduction tree. Therefore, the size of the output
    file $\f{c}$ does not exceed $\mathit{InIdle}$, that is, the size
    of all possible input files stored in memory right before the
    start of the processing of inner node $c$, but not used at that
    time, because this includes all the input files of inner node
    $c$. Also, the total size of the output files of the processed
    inner nodes, $\mathit{Out}_\textsc{in}$, cannot exceed the total
    size of the input files of the processed inner nodes,
    $\mathit{In}_\textsc{in}$.  Therefore, $M_{\mathit{new}} =
    \mathit{In}_\textsc{in} + \mathit{Out}_\textsc{in} +
    \mathit{Out}_\textsc{lf} + \mathit{InIdle} + \f{c} \leq
    \mathit{In}_\textsc{in} + \mathit{Out}_\textsc{in} +
    \mathit{Out}_\textsc{lf} + 2 \mathit{InIdle} \leq 2
    \mathit{In}_\textsc{in} + \mathit{Out}_\textsc{lf} + 2
    \mathit{InIdle} \leq 2( \mathit{In}_\textsc{in} +
    \mathit{Out}_\textsc{lf} + \mathit{InIdle})$.

    So the new memory requirement $M_{\mathit{new}}$ is not greater
    than twice the memory occupied by all \textit{input} files and all
    \emph{output} files of leaf nodes. Because the tree is by
    hypothesis a reduction tree, executing an inner node never
    increases the total size of all input files and all output files of
    leaves. This can only happen by starting a leaf, but that is
    not done if it would exceed the required memory $M$. Therefore,
    $\mathit{In}_\textsc{in} + \mathit{Out}_\textsc{lf} +
    \mathit{InIdle}$ never exceeds $M$ and $M_{\mathit{new}} \leq 2M$.
\end{compactenum}

We now prove that when the algorithm ends the entire input tree has
been processed. We reason by contradiction and assume that this is not
the case.  Ready inner nodes are processed without checking the amount
of available memory. Therefore, when the algorithm terminates without
having completed the processing of the tree, $\mathit{eventSet}$ is
empty but some leaves have not been processed. Then, let $l$ be the
first un-processed leaf, according to the order $O$. At the time
Algorithm~\ref{algo.listSchedulingMemLimit} terminates, it has
processed exactly the same leaves as the sequential algorithm when
it starts processing leaf $l$. Because $\mathit{eventSet}$ is empty,
there are no remaining ready inner nodes and no node is processed at
the time of the algorithm termination. Because of the hypothesis that
$T$ is a reduction tree, the amount of available memory when
Algorithm~\ref{algo.listSchedulingMemLimit} terminates is not smaller
than the amount of available memory under the sequential algorithm
right before it starts processing leaf $l$.  Because the sequential
algorithm can process the whole tree with a peak memory usage of
$M_{seq} \leq M$, the processing of leaf $l$ can be started by
Algorithm~\ref{algo.listSchedulingMemLimit}. This contradicts the
assumption of early termination.
\end{proof}

We define a variant \ParDeepestFirstMemLimitOptim
of \ParDeepestFirstMemLimit, and a variant \ParInnerFirstMemLimitOptim
of \ParInnerFirstMemLimit, by being more aggressive about starting
leaves. Instead of checking for the condition $M_{\mathit{used}} + \f{c} \leq M$
before starting a leaf node $c$ (line~\ref{algo.line.check} of
Algorithm~\ref{algo.listSchedulingMemLimit}), it is in fact sufficient
to check that $\mathit{In}_\textsc{in} +
\frac{1}{2}\mathit{Out}_\textsc{lf} + \mathit{InIdle} + \f{c} \leq M$
(using the notation of the proof of Theorem~\ref{theo.memLimitedAlgos}).
For Case~\eqref{thm6.case1} of the proof, one just needs to
remark that after leaf $c$ is started, $M_{\mathit{new}} =
\mathit{In}_\textsc{in} + \mathit{Out}_\textsc{in} +
\mathit{Out}_\textsc{lf} + \mathit{InIdle} + \f{c}$. Then, because the
tree is a reduction tree, $\mathit{Out}_\textsc{in} \leq
\mathit{In}_\textsc{in}$. Therefore, $M_{\mathit{new}} \leq
2\mathit{In}_\textsc{in} + \mathit{Out}_\textsc{lf} + \f{c}+
\mathit{InIdle} \leq 2\mathit{In}_\textsc{in} +
\mathit{Out}_\textsc{lf} + \f{c}+ 2\mathit{InIdle}$, which, in turn,
is no greater than $2M$ because of the new condition.
The modified condition has no impact on the study of Case~\eqref{thm6.case2}, because the inequality $\mathit{In}_\textsc{in} + \frac{1}{2}\mathit{Out}_\textsc{lf} + \mathit{InIdle} \leq M$ is sufficient to conclude that case.

\paragraph{Memory}
Theorem~\ref{theo.memLimitedAlgos} establishes that the peak memory
required by \ParInnerFirstMemLimit and \ParDeepestFirstMemLimit with
$p$ processors is at most twice that of their sequential execution
($p=1$) with the same order. It should be noted that this peak
requirement $M_{seq}$ does not correspond in general to the memory
requirement of an \textit{optimal} sequential algorithm. In
particular, the sequential execution of \ParInnerFirstMemLimit
corresponds to a postorder traversal, which is not optimal for all
instances, but generally provides good
results~\cite{ipdps-tree-traversal}. We propose to
use \ParInnerFirstMemLimit with a node order $O$ that corresponds to
an optimal sequential postorder, e.g., with \cite{Liu86}. The memory
requirement of the sequential \ParDeepestFirstMemLimit is unbounded
compared to the optimal sequential memory requirement, because the
same arguments apply as the ones discussed for \ParDeepestFirst in
Section~\ref{sec.parDeepestFirst}.

\paragraph{Makespan} We have already stated that the above heuristics
are $p$-approximation algorithms for makespan minimization. The
following lemma refines this result:
\begin{lemma}
  \ParInnerFirstMemLimit and \ParInnerFirstMemLimitOptim are both
  $p$-approximation algorithms for makespan minimization and this
  bound is tight.
\end{lemma}
\begin{proof}
  We establish this result by studying the tree in
  Figure~\ref{fig.mem-bounded-list-sched-heterogeneous}. This tree can
  be processed with a peak memory usage of $M$. We assume
  that \ParInnerFirstMemLimit is called with this memory limit.
  The key observation is that in any schedule, among the three
  descendants of an $a_i$ node, the nodes $c_i$ and $b_i$ must be
  processed before the $d_i$ node: otherwise, keeping in memory the
  output file of size $M/p$ of node $d_i$ makes it impossible to start
  processing the leaf node $c_i$ because of its output file of size
  $M$. And since under \ParInnerFirstMemLimit, leaf nodes are
  processed according to a postorder traversal, the processing of the
  subtrees is sequentialized and the overall processing takes time
  $p(2+k)$. On the other hand, with respect to the makespan, it would
  be better to first sequentially process in that order $c_1$, $b_1$,
  $c_2$, $b_2$, ..., $c_p$, and $b_p$, which would take a time $2p$,
  and then process in parallel the $d_i$'s for an overall makespan of
  $2p+k$. Hence, on this example, the approximation ratio of
  \ParInnerFirstMemLimit is no smaller than $\frac{p(2+k)}{2p+k}$
  which tends to $p$ when $k$ tends to infinity.
\end{proof}
We do not have a similar result for the memory limited deepest first
algorithms as already the sequential traversal with these algorithms (which
determines the given memory limit) can require significantly more memory
than a postorder traversal. For the example in
Figure~\ref{fig.mem-bounded-list-sched-heterogeneous}, the minimum sequential
memory for a deepest first traversal is equal to $pM$. The additional
memory buys a lot of freedom for  \ParDeepestFirstMemLimit and makes
the comparison harder.


%

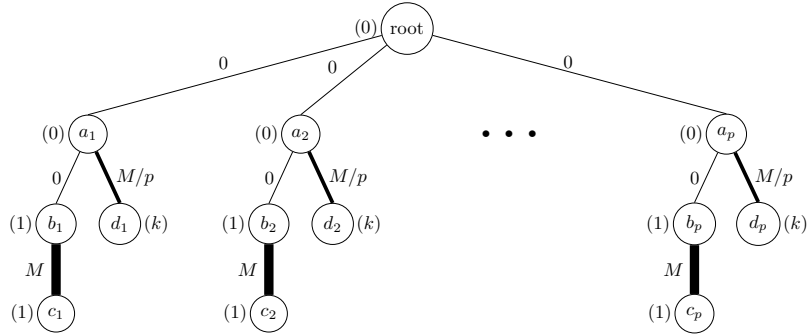
\begin{figure*}
  \centering
  \scalebox{0.7}{
  \begin{tikzpicture}[child anchor = north, label distance = -2pt]
    \node (root) [circle,draw,label=left:(0)]{root}
    [sibling distance=40mm, level distance=20mm]
    child{  node[circle, draw,label=left:(0)]{$a_1$}
      [sibling distance=12mm, level distance=17mm]
      child{node[circle, draw,label=left:(1)]{$b_1$}
        child{node[circle, draw,label=left:(1)]{$c_1$} edge from parent [line width=5pt] node [left] {$M$} } edge from parent node [left] {0}  }
      child{node[circle, draw,label=right:($k$)]{$d_1$} edge from parent [line width=2pt] node [right] {$M/p$} }
      edge from parent node [above left] {0}
    }
    child{  node[circle, draw,label=left:(0)]{$a_2$}
      [sibling distance=12mm, level distance=17mm]
      child{node[circle, draw,label=left:(1)]{$b_2$}
        child{node[circle, draw,label=left:(1)]{$c_2$} edge from parent [line width=5pt] node [left] {$M$} } edge from parent node [left] {0}  }
      child{node[circle, draw,label=right:($k$)]{$d_2$} edge from parent [line width=2pt] node [right] {$M/p$} }
      edge from parent node [above left] {0}
    }
    child{ node[circle, draw,draw=none, fill=white]{\Huge \bfseries \ldots } edge from parent [draw=none] }
    child{  node[circle, draw,label=left:(0)]{$a_p$}
      [sibling distance=12mm, level distance=17mm]
      child{node[circle, draw,label=left:(1)]{$b_p$}
        child{node[circle, draw,label=left:(1)]{$c_p$} edge from parent [line width=5pt] node [left] {$M$} } edge from parent node [left] {0}  }
      child{node[circle, draw,label=right:($k$)]{$d_p$} edge from parent [line width=2pt] node [right] {$M/p$} }
      edge from parent node [above left] {0}
    }
    ;
  \end{tikzpicture}}
\caption{Tree used to establish the worst-case performance of most
  memory bounded heuristics. Node labels in parentheses are processing
  times, edge labels are memory weights. All nodes have null-size
  processing files.}
  \label{fig.mem-bounded-list-sched-heterogeneous}
\end{figure*}

\subsubsection{Memory booking heuristic \BookingInnerFirst}
\label{sec:memorybooking}

The two heuristics described in the previous section satisfy an
achievable memory bound, $M$, in a relaxed way: the guarantee is that
they never use more than twice the memory limit. Here, we aim at
designing a heuristic that satisfies an achievable memory bound $M$ in
the strong sense: the heuristic never uses more than a memory of $M$.

To achieve such a goal, we want to ensure that whenever an inner node
$i$ becomes ready there is enough memory to process it. Therefore, we
book in advance some memory for its later processing. Our goal is to
book as little memory as possible, and to do so as late as possible.
The algorithm then relies on a sequential postorder schedule, denoted
$\PO$: for any node $k$ in the task graph, $\PO(k)$ denotes the step
at which node $k$ is executed under $\PO$. Let $j$ be the last child
of $i$ to be processed.  If the total size of the input files of $j$
is larger than (or equal to) $f_i$, then only that last child will
book some memory for node $i$. In this case (part of) the memory that
was used to store the input files of $j$ will be used for $f_i$.  If
the total size of the input files of $j$ is smaller than $f_i$, then
the second to last child of $i$ will also have to book some memory for
$f_i$, and so on. The following recursive formula states the amount of
memory $\mathit{Contrib}[j]$ a child $j$ has to book for its parent
$i$:
\begin{equation*}
  \mathit{Contrib}[j] = \min\left(
    \inputs{j},
    f_i - \sum_{\substack{j' \in \Children{i}\\ \PO(j')>\PO(j)}} \mathit{Contrib}[j']
  \right)
\end{equation*}

If $j$ is a leaf, it may also have to book some memory for its parent.
However, the behavior for leaves is quite different than for inner
nodes. A leaf node cannot transfer some of the memory used for its
input files (because it does not have any) to its parent for its
parent output file. Therefore, the memory booked by a leaf node may
not be available at the time of the booking. However, this memory will
eventually become available (after some inner nodes are processed);
booking the memory prevents the algorithm from starting the next leaf
if it would use too much memory: this ensures that the algorithm
completes the processing without violating the memory bound. The
contribution of a leaf $j$ for its parent $i$ is:
\begin{equation*}
  \mathit{Contrib}[j] = f_i - \sum_{\substack{j' \in \Children{i}\\ \PO(j')>\PO(j)}} \mathit{Contrib}[j']
\end{equation*}
Note that the value of $\mathit{Contrib}$ for each node can be
computed before starting the algorithm, in a simple tree
traversal. Using these formulas, we are able to guarantee that enough
memory is booked for each inner node $i$:
$$
\sum_{j\in \Children{i}} \mathit{Contrib}[j] = f_i.
$$
Using these definitions, we design a new heuristic,
\BookingInnerFirst, which is described in
Algorithm~\ref{alg.booking}. In this algorithm, $\mathit{Booked}[i]$
denotes the amount of memory currently booked for the processing of an
inner node $i$. We make use of a new notation: we denote by
$\Ancestors{i}$ the set of nodes on the path from $i$ to the root node
(excluding $i$ itself), that is, all ancestors of
$i$. 

Note that, contrarily to \ParInnerFirstMemLimit, \BookingInnerFirst
does not guarantee that there is always enough memory available to
process an inner node $i$ as soon as it becomes ready. This is why
Lemma~\ref{le.booking.inner} only guarantees that an inner node $i$
will \emph{eventually} be processed if a leaf $j$ with $\PO(j)>\PO(i)$
is started by \BookingInnerFirst.


\begin{algorithm}[tb]
  \DontPrintSemicolon
  \caption{\BookingInnerFirst($T$, $p$, $\PO$, $M$)}
  \label{alg.booking}
  \KwIn{tree $T$, number of processor $p$, postorder $\PO$, memory limit $M$ (not
    smaller than the peak memory of the sequential
    traversal defined by $\PO$)}
  \lForEach{node $i$}{
    $\mathit{Booked}[i] \gets 0$\;
  }
  $M_{\mathit{used}} \gets 0$\;
  \While{the whole tree is not processed}{
    Wait for an event (task finish time or starting point of the algorithm)\;
    \ForEach{finished non-leaf node $j$ with parent $i$}{
      $M_{\mathit{used}} \gets M_{\mathit{used}} - \inputs{j}$\;
      $\mathit{Booked}[i] \gets \mathit{Booked}[i] + \mathit{Contrib}[j]$\;
    }
    $\Done \leftarrow$ set of the new ready nodes\;
    Insert nodes from \Done in $\mathit{PQ}$ according to order $\PO$\;
    $\mathit{WaitForNextTermination} \gets \false$\;
    \While{$\mathit{WaitForNextTermination} = \false$ \textbf{and} there is an available processor $P_u$ \textbf{and} $\mathit{PQ}$ is not empty}{
      $j \gets pop(\mathit{PQ})$\;
      \If{$j$ is an inner node and $M_{\mathit{used}} + f_j \leq M$}{
        $M_{\mathit{used}} \gets M_{\mathit{used}} + f_j$\;
        $\mathit{Booked}[j] \gets 0$\;
        Make $P_u$ process $j$\;
      }
      \ElseIf{$j$ is a leaf and $M_{\mathit{used}} + f_j + \sum_{k\notin \Ancestors{j}} \mathit{Booked}[k] \leq M$\label{alg.booking.line16}}{
        $M_{\mathit{used}} \gets M_{\mathit{used}} + f_j$\;
          $\mathit{Booked}[\text{parent of }j] \gets
          \mathit{Booked}[\text{parent of }j] + \mathit{Contrib}[j]$\;
        Make $P_u$ process $j$\;
      }
      \Else{
        $push(j,\mathit{PQ})$\;
        $\mathit{WaitForNextTermination} \gets \true$\;
      }
    }
}
\end{algorithm}

\begin{lemma}
  Consider any inner node $i$. If some leaf $j$ with $\PO(j)>\PO(i)$ has
  been started by Algorithm~\ref{alg.booking}, then at some point,
  there will be enough memory to process $i$.
  \label{le.booking.inner}
\end{lemma}


\begin{proof}
  By contradiction, assume that an available inner node $i$ can never
  be processed because of memory constraints, that is,
  Algorithm~\ref{alg.booking} stops without processing $i$, and some
  leaf $j$ with $\PO(j)>\PO(i)$ has been started (in case of several
  such leaves, we consider the one with largest $\PO(j)$). Note
  that $i$ cannot be a parent of $j$ (otherwise we would have $\PO(i)
  > \PO(j)$). We consider the amount $A=M-M_{\mathit{used}} -
  \sum_{k \notin \Ancestors{j}} \mathit{Booked}[k]$ and its
  evolution. Before starting $j$, we check that $A \geq f_j$. When
  starting $j$, the amount of available memory is decreased by $f_j$
  and, thus, we have $A \geq 0$. The following events may happen after
  the beginning of $j$:
  \begin{compactitem}
  \item Some inner node $u$ not in $\Ancestors{j}$ is terminated. Let
    us call $v$ its parent. When $u$ completes,
    $M_{\mathit{used}}$ decreases by \inputs{u}, while
    $\mathit{Booked}[v]$ increases by $\mathit{Contrib}[u] \leq
    \inputs{u}$. Thus, $A$ does not decrease.
  \item Some inner node $k$ not in $\Ancestors{j}$ is started. In that
    case, the booked memory $f_k$ is traded for used memory, and the
    total memory amount $A$ is preserved.
  \item An inner node $u$ in $\Ancestors{j}$ is started. In this
    case, the amount of available memory may temporarily
    decrease. However, because of the reduction property, the amount
    of memory freed when $u$ completes is not smaller than the amount
    of additional memory temporarily used for the processing of $u$.
  \item A leaf node has completed: this modifies neither the amount of
    available or booked memory and, so, $A$ is left unchanged.
  \end{compactitem}
  Therefore, when the algorithm stops with $i$ available, $A\geq
  0$. Thus, $M-M_{\mathit{used}}\geq \mathit{Booked}[i] = f_i$:
  there is enough memory to process $i$.
\end{proof}




Using the previous lemma, we now prove Algorithm \BookingInnerFirst.

\begin{theorem}
  \BookingInnerFirst called with a postorder $\PO$ and a memory bound
  $M$ processes the whole tree with memory $M$ if $M$ is not smaller
  than the peak memory of the sequential traversal defined by $\PO$.
\end{theorem}

\begin{proof}
  By contradiction, assume that the algorithm stops while some nodes
  are unprocessed. We consider two cases:
  \begin{itemize}
  \item There is at least one available unprocessed inner node $i$ (if
    there are several, we choose the one with the smallest $\PO$
    value).  Consider the step $\PO(i)$ when this node $i$ is
    processed in the sequential postorder schedule. At this time, the
    set $\mathcal{S}$ of the leaves processed by the sequential
    postorder is exactly the set of the leaves $j$ such that
    $\PO(j)<\PO(i)$. Thanks to Lemma~\ref{le.booking.inner}, we know
    that \BookingInnerFirst has not processed any leaf $j$
    with $\PO(j)>\PO(i)$. Therefore, the set of the leaves processed
    by \BookingInnerFirst is a subset of $\mathcal{S}$.  Node
    $i$ being available, \BookingInnerFirst has processed all
    the leaves in the subtree $\mathit{ST}$ rooted at
    $i$. \BookingInnerFirst cannot start a leaf $k$ if a leaf
    $j$ with $\PO(j) < \PO(k)$ has not been started. Therefore all
    leaves that precedes in the postorder the leaves of $\mathit{ST}$
    have also been processed by \BookingInnerFirst. By
    definition of a postorder, there is no leaf that does not belong
    to $\mathit{ST}$ that is scheduled after the first of the leaf of
    $\mathit{ST}$ and before $i$. Therefore,
    \BookingInnerFirst has processed the exact same set of
    leaves than the sequential postorder at step $\PO(i)$.  We now
    prove that the same set of inner nodes have been processed by both
    algorithms:
    \begin{itemize}
    \item Assume that an inner node $k$ has been processed by
      \BookingInnerFirst but not by the sequential postorder
      at time $\PO(i)$. Since $k$ has not yet been processed by the
      sequential postorder, $\PO(k)>\PO(i)$. Since no leaf $j$ with
      $\PO(j)>\PO(i)$ has been processed by
      \BookingInnerFirst, since $\PO(k)>\PO(i)$, and since $\PO$ is a
      postorder, then $k$ can only be
      a parent of $i$, which contradicts the fact that $i$ is not
      processed.
    \item Assume that an inner node $k$ has been processed by the
      sequential postorder at time $\PO(i)$ but not by
      \BookingInnerFirst. Since it has been processed before
      $i$ in the sequential postorder, $\PO(k)<\PO(i)$. This node, or
      one of its inner node predecessor, must be available in
      \BookingInnerFirst (note that it cannot be a leaf,
      since all leaves with $\PO$ values smaller than $\PO(i)$ are
      already processed). This contradicts the fact that $i$ is the
      available inner node with smallest $\PO$ value.
    \end{itemize}
    Thus, there is no difference in the state of the sequential
    postorder when it starts $i$ and \BookingInnerFirst when
    it stops, including in the amount of available memory.  This
    contradicts the fact that $i$ cannot be started because of memory
    issues.

  \item There is no unprocessed available inner node. Thus, some
    leaf is available and cannot be processed. Let $j$ be the first
    of these leaves according to $\PO$. None of the inner nodes for
    which some memory has been booked is available and, thus, they are
    all parents of $j$ (because $\PO$ is a postorder and because the
    processing of all the leaves that precede $j$ in $\PO$ has been
    completed). Thus, the memory condition which prevents $j$ to be
    executed can be rewritten: $M-M_{\mathit{used}} <
    f_j$. However, since no inner node is available, this is the same
    situation as right before $j$ is processed in the sequential postorder,
    which contradicts the fact that $j$ can be processed in the
    sequential postorder.
  \end{itemize}
\end{proof}

\begin{lemma}
  \BookingInnerFirst is a $p$-approximation algorithm for makespan
  minimization, and this bound is tight.
\end{lemma}
This result is proved following the exact same arguments than for the
bound on the performance of \ParInnerFirstMemLimit, including the tree
on Figure~\ref{fig.mem-bounded-list-sched-heterogeneous}.

\paragraph{Complexity}

Algorithm~\ref{alg.booking} can be implemented with the same
complexity as the other heuristics, namely $O(n\log(n))$ (which comes
from the management of the $PQ$ queue). The only operations added to
this algorithm which could increase this complexity is the test
executed on line~\ref{alg.booking.line16} to make sure that a new leaf
can be started, that is, the computation of $\sum_{k\notin \Ancestors{j}}
\mathit{Booked}[k]$ for each leaf might take $O(n^2)$ time if not done
carefully. However, it is possible to avoid recomputing the values too
many times. We first remark the following property: when leaf $j$ has
not been started,
$$
\sum_{k\notin \Ancestors{j}} \mathit{Booked}[k] = \sum_{\PO(k)< \PO(j)} \mathit{Booked}[k].
$$

Indeed, if leaf $j$ has not been started, the postorder property
ensures that any $k \notin \Ancestors{j}$ with $\PO(k)\geq\PO(j)$ has
$\mathit{Booked}[k] = 0$, because none of its children have started their
execution.

For an efficient implementation, we keep a record of $R=\sum_{\PO(k)<
  \PO(j)} \mathit{Booked}[k]$ for the leaf $j$ which was tested on the
last execution of line~\ref{alg.booking.line16}. To keep this record,
it is enough to
\begin{itemize}
\item decrease $R$ by $\mathit{Booked}[i]$ each time an inner node $i$
  with $\PO(i)<\PO(j)$ begins execution,
\item increase $R$ by $\mathit{Contrib}[i]$ each time an inner node
  $i$ with $\PO(i)<\PO(j)$ finishes,
\item and increase $R$ by $\sum_{\PO(j) \leq \PO(k)< \PO(j')}
  \mathit{Booked}[k]$ if a new leaf $j'$ is being considered in the
  test of line~\ref{alg.booking.line16}.
\end{itemize}

In total, the number of updates to $R$ over the course of the whole
algorithm is bounded by $2n$: each $\mathit{Contrib}$ value is added
at most once to $R$, and each $\mathit{Booked}$ value is subtracted
at most once. Furthermore, the cost of computing the
sums $\sum_{\PO(j) \leq \PO(k)< \PO(j')} \mathit{Booked}[k]$ is also
bounded by $n$ since each node is considered only once. Hence, these
updates do not increase the total complexity of $O(n\log(n))$ of the
whole algorithm.

\section{Experimental validation}
\label{sec.experiments}

In this section, we experimentally compare the heuristics proposed in
the previous section, and we compare their performance to lower bounds.

\subsection{Setup}

All heuristics have been implemented in C. Special care has been
devoted to the implementation to avoid complexity issues. Especially,
priority queues have been implemented using binary heap to allow for
$O(\log{n})$ insertion and minimum extraction.
We have also implemented Liu's algorithm~\cite{Liu87} to obtain the
minimum sequential memory peak, which is used as a lower bound on memory
for comparing the heuristics.

\subsection{Data set}

The data set contains assembly trees of a set of sparse matrices
obtained from the University of Florida Sparse Matrix Collection
(\url{http://www.cise.ufl.edu/research/sparse/ matrices/}). The chosen
matrices satisfy the following assertions: not binary, not
corresponding to a graph, square, having a symmetric pattern, a number
of rows between 20,000 and 2,000,000, a number of nonzeros per row at
least equal to 2.5, and a number of nonzeros at most equal to
5,000,000; and each chosen matrix 
has the largest number of nonzeros among the matrices in its group
satisfying the previous assertions. With these criteria we
automatically select a set of medium to large matrices from different
application domains with nontrivial number of nonzeros.
 At the time of testing there were 76 matrices
satisfying these properties.  We first order the matrices using
MeTiS~\cite{kaku:98:metis} (through the MeshPart toolbox~\cite{gimt:98})
and {\tt amd} (available in Matlab), and then build the corresponding
elimination trees using the {\tt symbfact} routine of Matlab.  We also
perform a relaxed node amalgamation~\cite{liu92multifrontal} on these elimination trees to create
assembly trees. We have created a large set of instances by allowing
1, 2, 4, and 16 (if more than $1.6 \times 10^5$ nodes) relaxed amalgamations per
node.

%
%
%
%


At the end we compute memory weights and processing times to
accurately simulate the matrix factorization: we compute the memory
weight \n{i} of a node as $\eta^2 +2\eta(\mu-1)$, where $\eta$ is the
number of nodes amalgamated, and $\mu$ is the number of nonzeros in
the column of the Cholesky factor of the matrix which is associated
with the highest node (in the starting elimination tree); the
processing time \len{i} of a node is defined as $2/3\eta^3 +
\eta^2(\mu-1) + \eta (\mu-1)^2$ (these terms corresponds to one
gaussian elimination, two multiplications of a triangular
$\eta\times\eta$ matrix with a $\eta\times(\mu-1)$ matrix, and one
multiplication of a $(\mu-1)\times \eta$ matrix with a $\eta \times
(\mu-1)$ matrix).  The memory weights \f{i} of edges are computed as
$(\mu-1)^2$.

The resulting 608 trees contains from 2,000 to 1,000,000 nodes. Their
depth ranges from 12 to 70,000 and their maximum degree ranges from 2
to 175,000. Each heuristic is tested on each tree using $p=2$, 4, 8,
16, and 32 processors. Then the memory and makespan of the resulting
schedules are evaluated by simulating a parallel execution.

\subsection{Results for heuristics without memory bound}

\begin{table}[htbp]
    \resizebox{\textwidth}{!}{%
      \begin{tabular}[tabular]{|c||c|c|c||c|c|c|}
        \hline
        \multirow{2}{*}{Heuristic}  & \multirow{2}{*}{Best memory}  & Within 5\% of & Normalized  & \multirow{2}{*}{Best makespan}  & Within 5\% of & Normalized\\
        &  &  best memory & memory  &  & best makespan & makespan\\
        \hline
        \ParSubtrees  & 81.1 \% & 85.2 \% & 2.34     & 0.2 \%  & 14.2 \%  & 1.40  \\
        \ParSubtreesOptim  & 49.9 \% & 65.6 \% & 2.46    & 1.1 \%  & 19.1 \%  & 1.33 \\
        \ParInnerFirst  & 19.1 \% & 26.2 \% & 3.79      & 37.2 \%  & 82.4 \%  & 1.07  \\
        \ParDeepestFirst  & 3.0 \% & 9.6 \% & 4.13     & 95.7 \%  & 99.9 \%  & 1.04 \\
        \hline
      \end{tabular}
    }
    \caption{Proportions of scenarii when heuristics reach best (or
      close to best) performance, and average deviations from optimal
      memory and best achieved makespan.\label{tab.cmp}}
\end{table}


The comparison of the first set of heuristics (without memory bounds)
is summarized in Table~\ref{tab.cmp}. It presents the fraction of the
cases where each heuristic reaches the best memory (respectively
makespan) among all heuristics, or when its memory (resp. makespan) is
within 5\% of the best one. It also shows the average normalized
memory and makespan. For each scenario (consisting in a tree and a
number of processors), the memory obtained by each heuristic is
normalized by the optimal (sequential) memory, and the makespan is
normalized using a classical lower bound, since makespan minimization
is NP-hard even without memory constraint. The lower bound is the maximum between
the total processing time of the tree divided by the number of
processors, and the maximum weighted critical
path.

Table~\ref{tab.cmp} shows that \ParSubtrees and \ParSubtreesOptim are
the best heuristics for memory minimization. On average they use less
than 2.5 times the amount of memory required by the optimal sequential
traversal, when \ParInnerFirst and \ParDeepestFirst respectively need
3.79 and 4.13 times this amount of memory. \ParInnerFirst
and \ParDeepestFirst perform best for makespan minimization, having
makespans very close on average to the best achieved ones, which is
consistent with their 2-approximation ratio for makespan minimization.
Furthermore, given the critical-path-oriented node ordering, we can
expect that \ParDeepestFirst makespan is close to
optimal. \ParDeepestFirst outperforms \ParInnerFirst for makespan
minimization, at the cost of a noticeable increase in memory.
\ParSubtrees and \ParSubtreesOptim may be better trade-offs, since
they use (on average) almost only half the memory of \ParDeepestFirst
for at most a 35\% increase in makespan.

\begin{figure}[htbp]
  \centering
  \scalebox{0.6}{\input{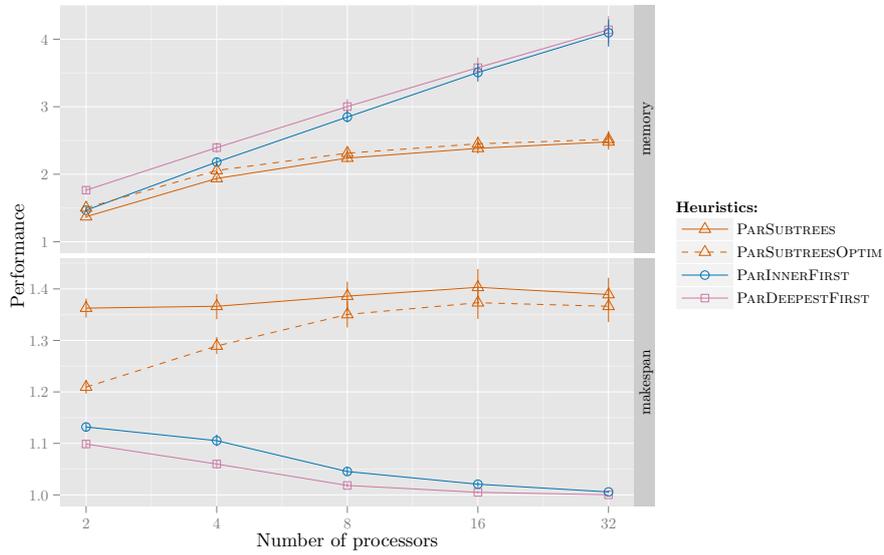}}
  \caption{Performance (makespan and memory) to the respective lower
    bounds for the first set of heuristics. Vertical bars represents
    confidence intervals when we exclude trees with extreme
    performance.}
  \label{fig.simu.unbounded}
\end{figure}

\begin{figure}[htbp]
  \centering
  \scalebox{0.6}{\input{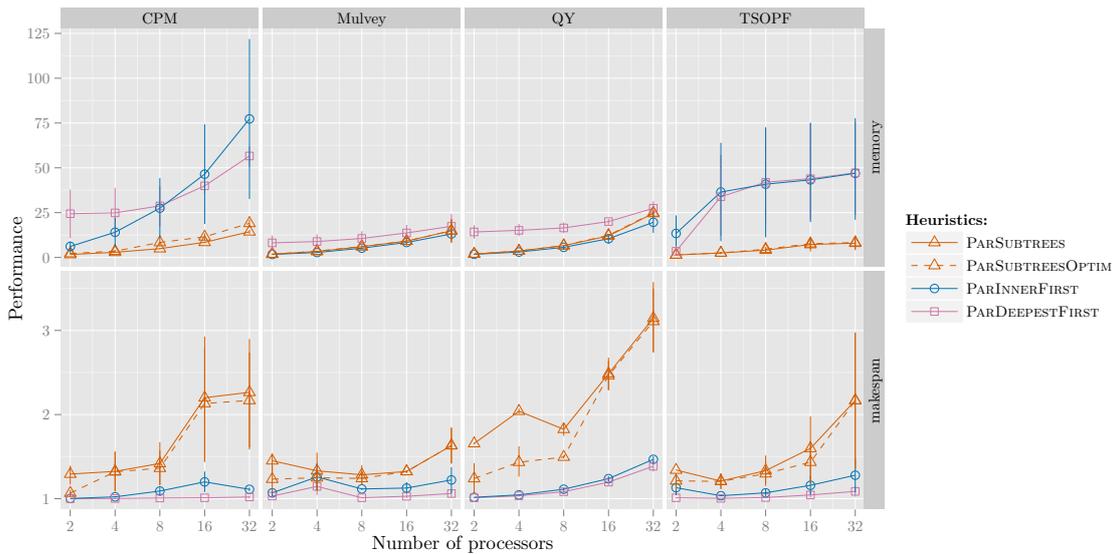}}
  \caption{Performance (makespan and memory) to the respective lower
    bounds for the first set of heuristics, on four specific classes
    of trees which show specific behavior. Vertical bars represents
    confidence intervals.}
  \label{fig.simu.unbounded.outliers}
\end{figure}

Figure~\ref{fig.simu.unbounded} presents the evolution of the
performance of these heuristics with the number of processors. On this
figure, we plot the results for all 608 trees except 76 of them, for
which the results are so different that it does not make sense to
compute average values anymore. These outliers belong to four
different classes of applications, and the specific results for these
graphs are shown on Figure~\ref{fig.simu.unbounded.outliers}.
Figure~\ref{fig.simu.unbounded} shows that \ParDeepestFirst
and \ParInnerFirst have a similar performance evolution, just
like \ParSubtrees and \ParSubtreesOptim. The performance gap between
these two groups, both for memory and makespan, increases with the
number of processors. With a large number of
processors, \ParDeepestFirst and \ParInnerFirst are able to decrease
the normalized makespan (at the cost of an increase of memory),
while \ParSubtrees has a an almost constant normalized makespan with
the number of processor.

Despite the very different values for makespan and memory utilization,
and a much higher variability, the results for the outliers give the
same conclusions about the relative performance of the
heuristics. Furthermore, this graph also exhibits the absence of
approximation ratios for \ParDeepestFirst and \ParInnerFirst for
memory minmization. Indeed, even though the trees used in this set are
taken from real-life applications, in contrast with the carefully
crafted counter-examples of Section~\ref{sec.heuristics}, the memory
usage of \ParDeepestFirst and \ParInnerFirst on those trees can reach
up to 100 times the optimal memory usage.

\subsection{Results for memory-bounded heuristics}

In addition to the previous heuristics, we also test the memory-bounded
heuristics. Since they can be applied only to reduction
trees with null processing sizes, we transform the trees used in the
previous tests into reduction trees as explained in Section~\ref{sec:reductiontrees}. For a given
scenario (tree, number of processors, memory bound), the memory obtain
by each heuristic is normalized by the optimal memory on the original
tree (not the reduction one).  Thus, the normalized memory represents
the actual memory used by the heuristic compared to the one of a
sequential processing. In particular, this allows a fair comparison
between memory-bounded heuristics and the previous unbounded
heuristics.


\begin{sidewaysfigure}
  \centering
  \scalebox{0.50}{\input{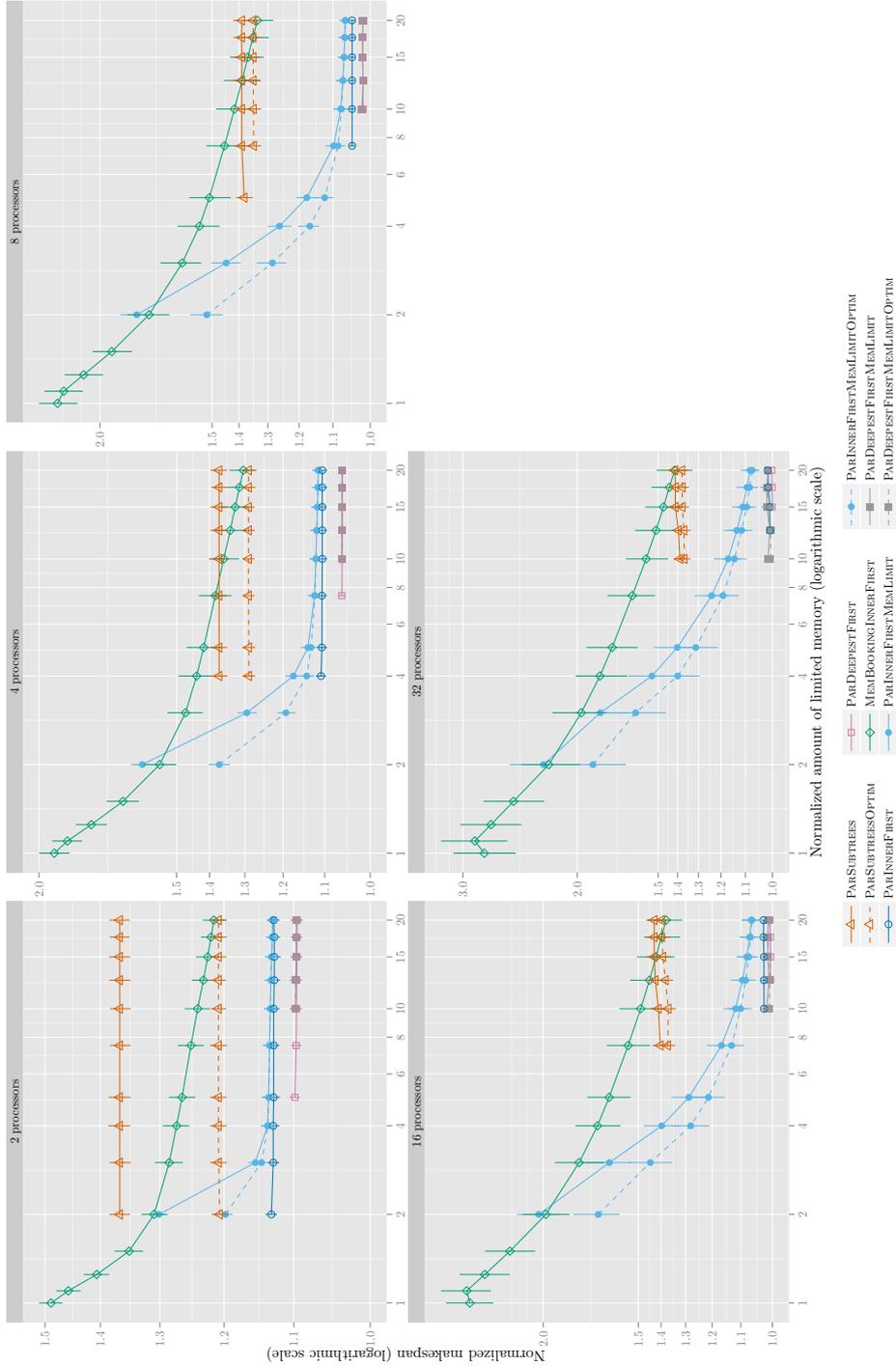}}
  \caption{Memory and makespan performance of memory-bounded
    heuristics (logarithmic scale).}
  \label{fig.simu.bounded}
\end{sidewaysfigure}

In order to compare the memory bounded heuristics, we have applied
them on the previous data-set, using various memory bounds. For each
tree, we first compute the minimum sequential memory $M_{\mathit{seq}}$ obtained by a
postorder traversal of the original tree. Then, each heuristic is
tested on the corresponding tree with a memory bound $B=x  M_{\mathit{seq}}$ for various ratios
$x\geq 1$. Sometimes, the heuristic cannot run because the amount of
available memory is too small. This is explained by the following
factors:
\begin{compactitem}
\item The memory-bounded heuristics use a reduction tree which may
  well need more memory than the original tree. In general, however, the
  transformation from original tree does not significantly increase
  the minimum amount of memory needed to process the tree.
\item \ParInnerFirstMemLimit has a minimum memory guarantee which is twice the
  sequential memory of a postorder traversal, thus it cannot run with
  a memory smaller than $2 M_{\mathit{seq}}$.
\item \ParDeepestFirstMemLimit has a minimum memory guarantee which is twice
  the sequential memory of a deepest first sequential traversal of the
  tree. A deepest first traversal uses much more memory than a
  postorder traversal, and thus \ParDeepestFirstMemLimit needs much more
  memory than \ParInnerFirstMemLimit to process a tree.
\end{compactitem}
Figure~\ref{fig.simu.bounded} presents the results of these
simulations. On this figure, points are shown only when a heuristic
succeeds in more than 95\% of the cases. The intuition is that with a
success rate larger than 95\%, the heuristic is presumably useful for
this ratio. This figure shows that when the memory is very limited
($B<2 M_{\mathit{seq}}$), \BookingInnerFirst is the only heuristic
that can be run, and it achieves reasonable makespans. For a less
strict memory bound ($2 M_{\mathit{seq}} \leq B< 5 M_{\mathit{seq}}$
or $2 M_{\mathit{seq}} \leq B< 10 M_{\mathit{seq}}$ depending on the
number of processors), \ParInnerFirstMemLimit is able to process the
tree, and achieves better makespans, especially when B is
large. Finally, when memory is abundant, \ParDeepestFirstMemLimit is
the best among all heuristics. On this figure, we also tested the two
variants
\ParInnerFirstMemLimitOptim and \ParDeepestFirstMemLimitOptim
presented in Section~\ref{sec:boundedlistsched} that are more
aggressive when starting leaves, but with the same memory guarantee
as \ParInnerFirstMemLimit and \ParDeepestFirstMemLimit. We see that
these strategies are able to better reduce the makespan in the case of
a very limited memory ($B$ close to $2 M_{\mathit{seq}}$).


\begin{figure}[htbp]
  \centering
  \scalebox{0.7}{
    \input{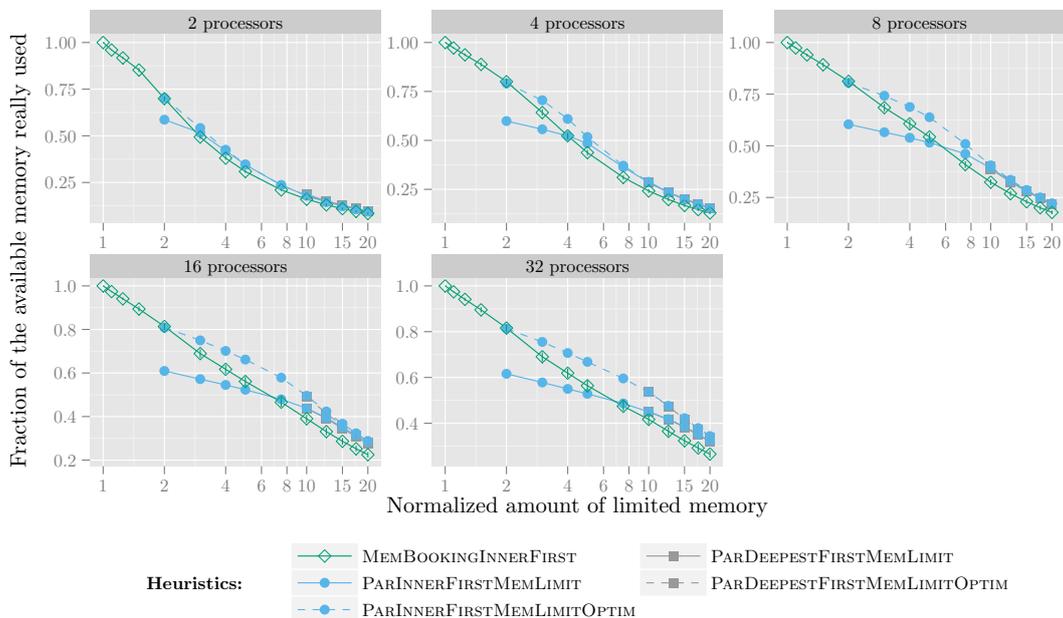}
  }
  \caption{Real use of limited memory for memory-bounded
    heuristics.}
  \label{fig.bounded-mem-usage}
\end{figure}

Finally, Figure~\ref{fig.bounded-mem-usage} shows the ability of
memory bounded heuristics to make use of the limited amount of
available memory. On this figure, points corresponding
to \ParDeepestFirstMemLimit
(respectively \ParDeepestFirstMemLimitOptim) are hardly
distinguishable from \ParInnerFirstMemLimit
(resp. \ParInnerFirstMemLimitOptim).  We notice that
\BookingInnerFirst is able to fully use the very limited amount of
memory when $B$ is close to $M_{\mathit{seq}}$. The good use of memory
is directly correlated with good makespan performance: for a given
amount of bounded memory, heuristics giving best makespans are the
ones that uses the largest fraction of available memory. Especially,
we can see that \ParInnerFirstMemLimitOptim
and \ParDeepestFirstMemLimitOptim are able to use much more memory
than their non-optimized counterpart, especially when memory is very
limited.

\section{Conclusion}

In this study we have investigated the scheduling of tree-shaped
task graphs onto multiple processors under a given memory limit
and with the objective to minimize the makespan. We started by
showing that the parallel version of the pebble
game on trees is NP-complete, hence stressing the negative impact of
the memory constraints on the complexity of the problem. More
importantly, we have proved that there does not exist any algorithm
that is simultaneously a constant-ratio approximation algorithm for
both makespan minimization and peak memory usage minimization when
scheduling tree-shaped task graphs. We have also established bounds on
the achievable approximation ratios for makespan and memory when the
number of processors is fixed. Based on these complexity results, we
have then designed a series of practical heuristics; some of these
heuristics are guaranteed to keep the memory under a given memory limit.
Finally, we have assessed the
performance of our heuristics using real task graphs arising from
sparse matrices computation. These simulations demonstrated that the
different heuristics achieve different trade-offs between the
minimization of peak memory usage and makespan; hence, the set of
designed heuristics provide an efficient solution for each situation.

This work represents an important step towards a comprehensive
theoretical analysis of memory/makespan minimization for applications
organized as trees of tasks, as it provides both complexity results
and memory bounded heuristics. Multifrontal sparse matrix
factorization is an important application for this work, and a good
incentive to refine the computation model. In a second step, we should
consider trees of parallel tasks rather than of pure sequential tasks,
as the computations corresponding to large tasks (at the top of the
tree) are usually distributed across processors. Of course, one would need
a proper computation model to derive relevant complexity
results. To get even closer to reality, one would also need to
consider distributed memory rather than shared memory, or a mix of
both. Hence, many important but challenging findings remain to be
done.

\section*{Acknowledgement}

We thank Bora Uçar for his help in creating and managing the data set
used in the experiments.  We gratefully acknowledge that this work is
partially supported by the Marsden Fund Council from Government
funding, Grant 9073-3624767, administered by the Royal Society of New
Zealand, and by the ANR {\em RESCUE} and {\em SOLHAR} projects, funded
by the French National Research Agency.

\bibliographystyle{plain}
\bibliography{biblio}


\end{document}